\documentclass[11pt]{article}
\usepackage{mathptm}
\usepackage{subfigure}
\usepackage{graphicx}
\usepackage{amsmath}
\usepackage{amsfonts}
\usepackage{amssymb}
\usepackage{epstopdf}
\usepackage{vertbars}
\usepackage{rotating} 
\newtheorem{theorem}{Theorem}
\newtheorem{definition}{Definition}

\newtheorem{example}{Example}

\newenvironment{proof}[1][Proof]{\textbf{#1.} }{\ \rule{0.5em}{0.5em}}

   {\vspace*{-5pt}
    \begin{enumerate}\itemsep=-4pt}%
   {\end{enumerate}
    \vspace*{-4pt}}

   {\vspace*{-5pt}
    \begin{itemize}\itemsep=-4pt}%
   {\end{itemize}
    \vspace*{-4pt}}

   {\vspace*{-5pt}
    \begin{description}\itemsep=-4pt}%
   {\end{description}
    \vspace*{-4pt}}
    
\begin{document}

\title{Compact DSOP and partial DSOP Forms}

\author{Anna Bernasconi\thanks{Department of Informatics, Universit\`a di Pisa, Italy. {\tt \{annab, luccio, pagli\}@di.unipi.it}} \and Valentina Ciriani\thanks{Department of Information Technology, Universit\`a degli Studi di Milano, Italy. {\tt valentina.ciriani@unimi.it}}  \and Fabrizio Luccio\footnotemark[1]  \and Linda Pagli\footnotemark[1]}

\maketitle

\begin{abstract}

Given a Boolean function $f$ on $n$ variables, a {\em Disjoint Sum-of-Products (DSOP)} of $f$ is a set of products (ANDs) of subsets of literals whose sum (OR) equals $f$, such that no two products cover the same minterm of $f$. 
DSOP forms are a special instance of {\em partial DSOPs}, i.e. the general case where a subset of minterms must be covered exactly once and the other minterms (typically corresponding to don't care conditions of $f$) can be covered any number of times.
We discuss finding DSOPs and partial DSOP with a minimal number of products, a problem theoretically connected with various properties of Boolean functions and practically relevant in the synthesis of digital circuits. Finding an absolute minimum is hard, in fact we prove that the problem of absolute minimization of partial DSOPs is NP-hard. Therefore it is crucial to devise a polynomial time heuristic that compares favorably with the known minimization tools. To this end we develop a further piece of theory starting from the definition of the {\em weight} of a product $p$ as a functions of the number of fragments induced on other cubes by the selection of $p$, and show how product weights can be exploited for building a class of minimization heuristics for DSOP and partial DSOP synthesis.
A set of experiments conducted on major benchmark functions show that our method, with a family of variants, always generates better results than the ones of previous heuristics, including the method based on a BDD representation of $f$. 
\end{abstract}

\section{Introduction}
\label{sec-intro}

Given a Boolean function $f$ on $n$ variables $x_1,x_2,...,x_n$ in $\cal{B}$$^n$, a {\em Disjoint Sum-of-Products (DSOP)} of $f$ is a set of products (ANDs) of subsets of literals whose sum (OR) equals $f$, such that no two products cover the same minterm of $f$. As each product is the mathematical expression for a cube in $\cal{B}$$^n$, a DSOP also represents a set of non intersecting cubes occupying the points of $\cal{B}$$^n$ in which $f=1$. In fact we shall indifferently refer to products or cubes, and apply algebraic or set operations to them. 
We are interested in 
finding a DSOP with a minimal number of products.

Besides its theoretical interest, DSOP minimization is relevant in the area of digital circuits for determining various properties of Boolean functions and for the synthesis of asynchronous circuits, as discussed for example in~\cite{FC99,LF09,LF09b,Sa93,TDM01}. DSOPs are indeed used as a starting point for the synthesis of {\em Exclusive-Or-Sum-Of-Products (ESOP)} forms, and for calculating the spectra of Boolean functions.

DSOP forms can be seen as a special case of {\em partial DSOPs} where a subset of minterms of a Boolean function must be covered exactly once, while other minterms can be covered more than once or not be covered at all. In particular this is the case where the points in the on set of a function are covered exactly once, while the points in the don't care set can be covered any number of times~\cite{LF09}. 

For speeding an otherwise exceedingly cumbersome process an absolute minimum in general is not sought for, rather heuristic strategies for cube selection have been proposed, working on explicit product expressions~\cite{BE10,FSC93,ST02}, or on a BDD representation of $f$~\cite{DHFD04,FD02}.  


After discussing  the complexity of DSOP and partial DSOP absolute minimization we propose a class of heuristic algorithms based on the new concept of ``cube weight'', and show that our results compare favorably with the ones of the other known heuristics.
The starting set of cubes is the one of a sum of product (SOP) found with standard heuristics. The SOP cubes may be eventually fragmented into non overlapping sub-cubes, giving rise to a largely unpredictable DSOP solution. The process may exhibit an exponential blow up in the number of fragments even dealing with theoretically minimal solutions, as for a function presented in~\cite{Sa95} where $|$SOP$|$ $=n/2$ and $|$DSOP$|$ $=2^{n/2}-1$ ($|$SOP$|$ and $|$DSOP$|$ denote the number of terms in the SOP and DSOP expression, respectively).

Another new characteristic of our heuristic is the idea of recomputing a SOP on the residual function at different possible stages of the disjoint minimization process, as a trade-off between quality of the result and computational time. We have observed experimentally that this strategy is crucial for obtaining compact DSOP forms.
For ease of presentation we start with DSOP synthesis and then extend the heuristics to the more general case of partial DSOP.



The paper is organized as follows. 
In the next Section~\ref{sec-NP} we discuss the complexity of absolute minimization of DSOP and partial DSOP forms proving that, for the latter, i.e. for the most general forms, the problem is NP-hard.
In Section~\ref{sec-weight} we define the weight of a product $p$ as a function of the number of fragments possibly induced on other cubes by the selection of $p$. In Section~\ref{sec-alg} we show how this weight can be exploited for building a class of minimization heuristics. Section~\ref{sec-alg2} extends our strategy to partial DSOP synthesis.
In Section~\ref{sec-exp} we present and discuss the computational results obtained by applying the proposed heuristic to the standard {\sc espresso} benchmark suite~\cite{Y91}, and  comparing these results with other published data. 
The paper is concluded in Section~\ref{sec-fut}.

\section{The complexity of DSOP minimization }
\label{sec-NP}

As it may be expected absolute DSOP minimization is a hard problem and absolute partial DSOP minimization may be at least as hard. 
Let us first recall some classical definitions.
In a Boolean space $\{ 0,1\} ^{n}$ described by $n$ variables $x_1$, $x_2$, $\ldots$, $x_n$, a {\em completely specified Boolean function} is a function $f: \{0,1\}^n \to \{0,1\}$, while Boolean a function $f$ is {\em partial} if $f: \{0,1\}^n \to \{0,1, -\}$. With usual terminology, a  {\em literal} $y_i$ is a variable $x_i$ in direct or complemented form, and {\em products} are ANDs of literals. A product $p$ is an {\em implicant} of the Boolean function $f$ if $\forall x\in \{ 0,1\} ^{n}, (p(x) =1) \Rightarrow (f(x)=1)$. An implicant $p$ of a function $f$ is a {\em prime implicant} if $p$ cannot be implied by a more general (i.e., with fewer literals) implicant of $f$.

Unlike SOPs, a DSOP composed of prime implicants only may not exist, as can be immediately seen considering a function with only three points in the on set, one adjacent to the other. Furthermore, DSOPs of prime implicants may exist but none of them may be minimal. For example the minimal DSOP cover of six implicants shown in Figure
1 contains the non prime implicant $x_1x_4\overline x_5\overline x_6x_7$ displayed in the sub-map $x_5x_6x_7=001$, which is covered by the prime implicant $x_1x_4\overline x_5\overline x_6$ spanning across the sub-maps $x_5x_6x_7=000$ and $x_5x_6x_7=001$. The reader may discover that there is one DSOP cover composed of seven prime implicants but not less (actually we could not construct an example with less than seven variables).  

\begin{figure*}\label{Fig:DSOP-nonP}
\begin{center}
\includegraphics[scale=0.28]{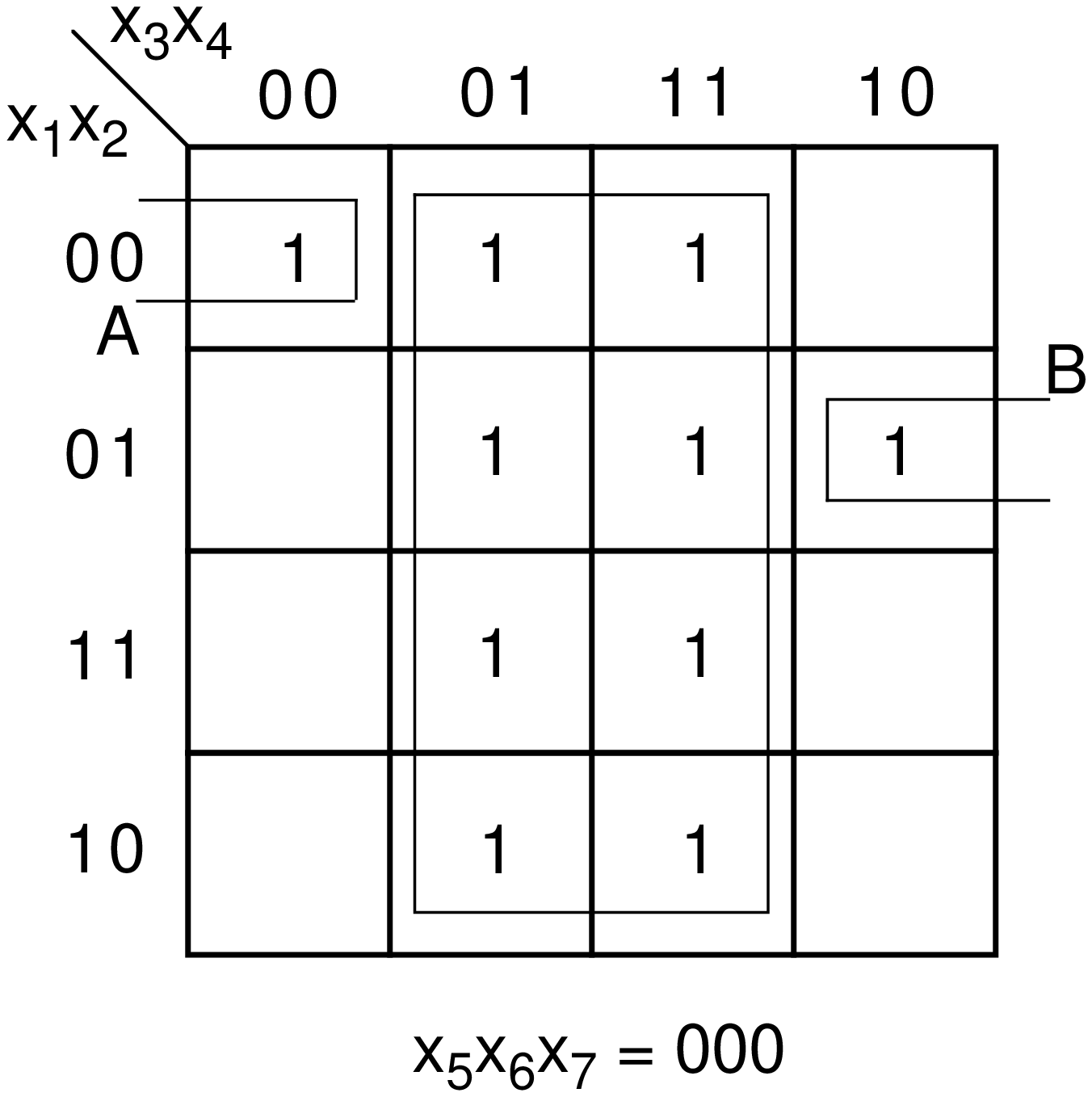}
\includegraphics[scale=0.28]{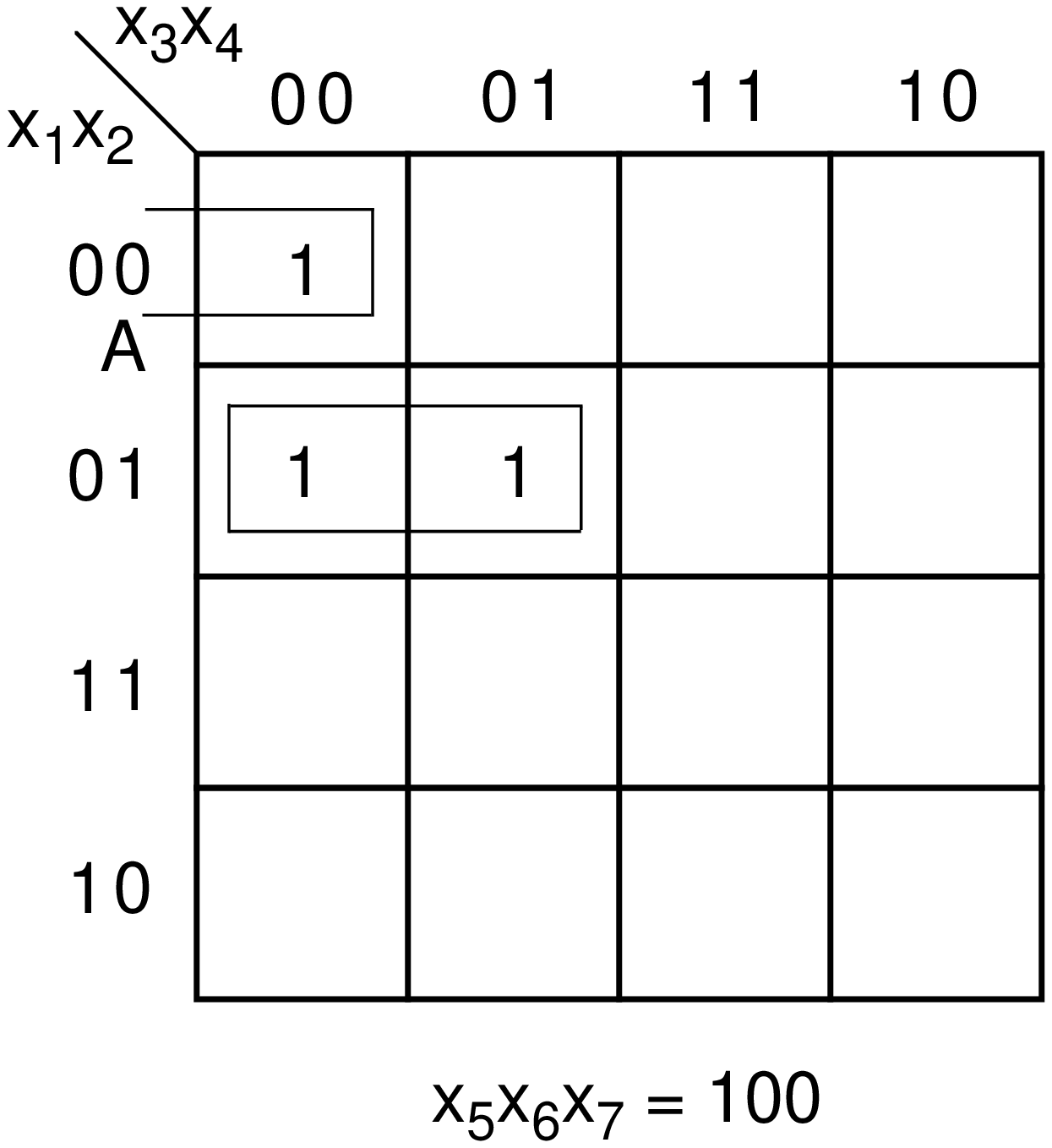}

\includegraphics[scale=0.28]{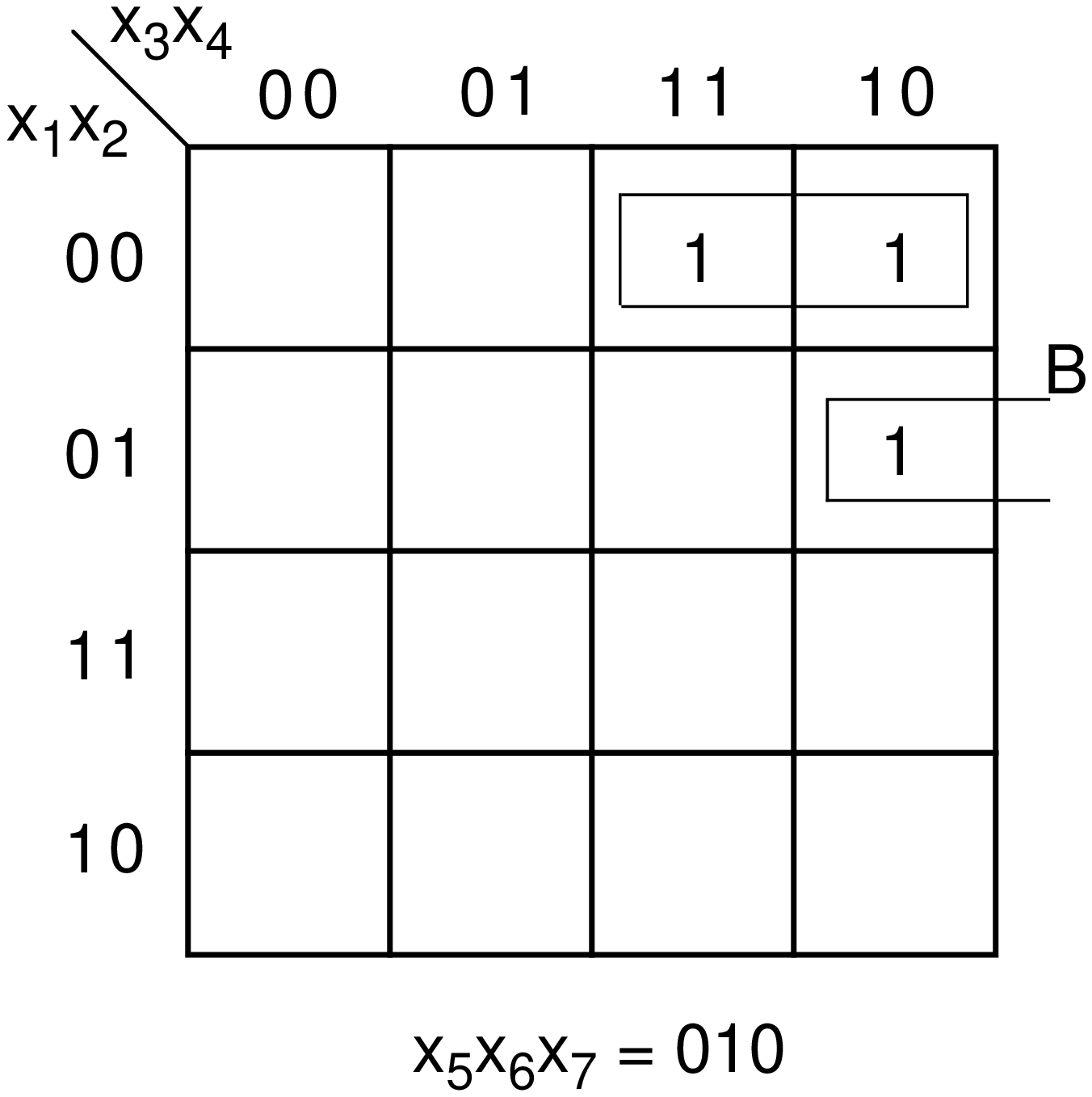}
\includegraphics[scale=0.28]{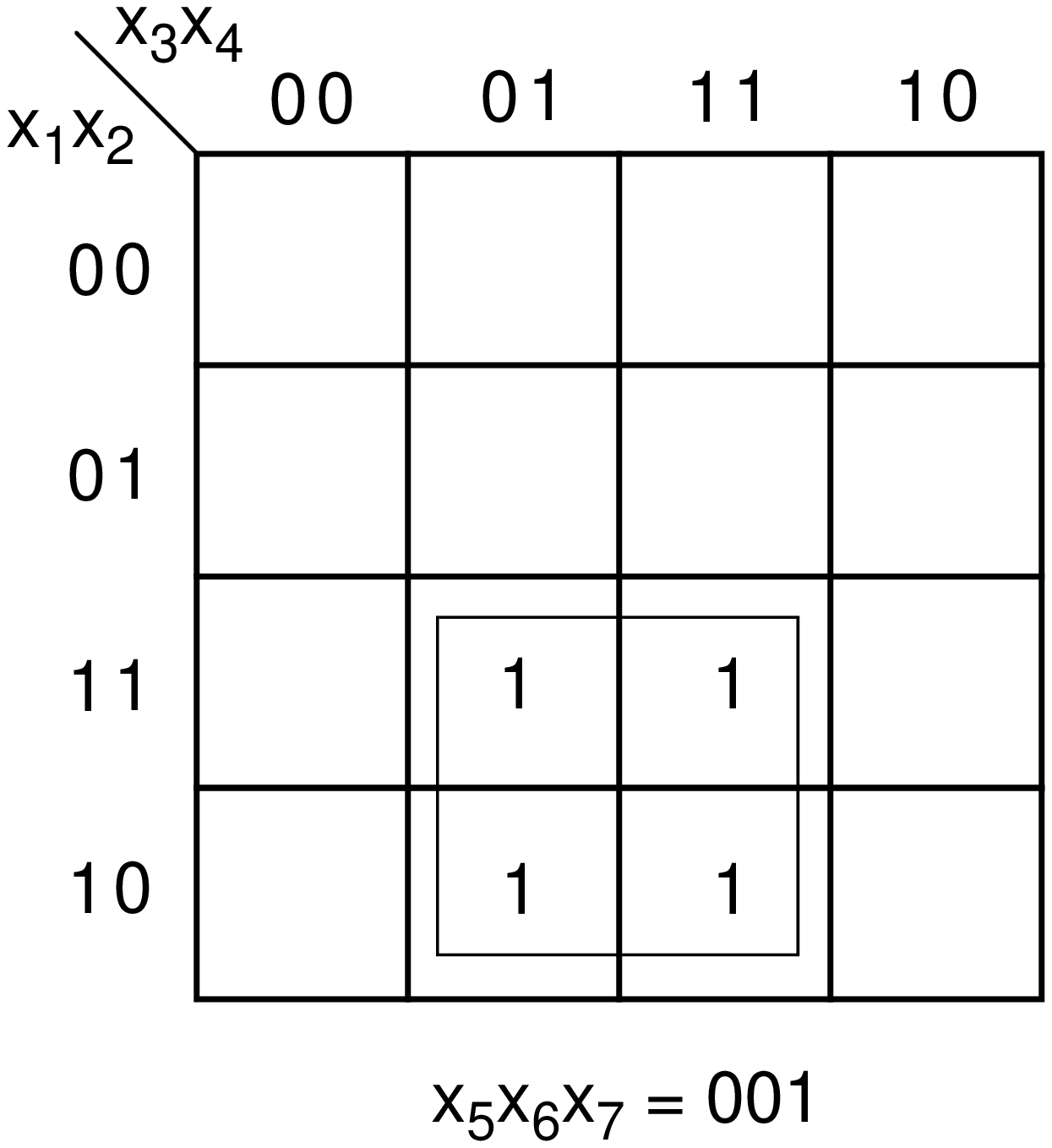}
\caption{A minimal DSOP in seven variables composed of six cubes, one of which is not prime (the Karnaugh maps for $x_5x_6x_7=011,101,110,111$ do not contain 1's and are not shown). A corresponding DSOP composed of prime cubes only includes at least seven of them.}
\end{center}
\end{figure*}

The above considerations show that, unlike in the SOP case, in DSOP minimization non prime implicants must also be considered. Theoretically this is not a major drawback as the generation of all implicants requires polynomial time in the size of the input (truth table of the function). The problem arises in the implicant selection phase where, as in the SOP case, a brute force enumerative selection requires exponential time in the worst case. It has been shown that SOP absolute minimization is as complex as set covering~\cite{GJ79,UVS06}. Similarly DSOP absolute minimization can be compared to the set partitioning (or minimal exact cover) problem.\footnote{The minimal exact cover problem is as follows: given a family of subsets $S$ of  a set $U$ and a positive integer $k$, is there a subset family $T\subseteq S$ such that the subsets in $T$ are $k$ in number, are disjoint, and their union is the entire set $U$? The minimal exact cover is NP-hard since it can be easily reduced by the ``exact cover problem'' introduced by Karp in 1972~\cite{K72} setting $k$ to the cardinality of $U$.} It is immediate that minimal exact cover is at least as hard as absolute DSOP minimization (solving the former problem efficiently would imply solving also the latter). Here we are not proving the reverse condition, rather we focalize on the most general problem of {\em partial} DSOP absolute minimization and prove that, in this version, the problem is NP-hard.
More precisely, we prove that the {\em decision version} of partial DSOP minimization is NP-complete. 
Let us formally define the problem.
\begin{quote} 
{\sc MIN Partial DSOP}

{\sc Input: } A partial Boolean function $f: \{0,1\}^n \to \{0,1, -\}$, specified by its on, off, and don't care set, and a positive integer $k$. 

{\sc Question: } Is there a partial DSOP, i.e., a sum of products covering exactly once the points of the on set, and any number of times the  points in the don't care set of $f$,  with at most $k$ products?
\end{quote}
This problem is in NP because given a ÒcandidateÓ partial DSOP with at most $k$ terms, one can determine whether it is a covering of $f$ satisfying the given requirements in time polynomial in the size of the input instance. In fact this simply requires evaluating the partial DSOP at all of the points in the on set of $f$ and checking that one and only one of its products takes the value 1. 

To prove the NP-completeness of {\sc MIN Partial DSOP}, we adapt to our problem the theory and the proofs developed in~\cite{AHMPS08}, where the authors proved that the decision version of finding the smallest SOP form consistent with a truth table is NP-complete, reducing from {\sc 3-Partite Set Cover}, instead of {\sc Circuit SAT} as done in~\cite{UVS06}. Moreover, they pointed out that the reduction would also work for {\sc 3D Matching}, which is precisely the NP-complete problem that we will reduce to {\sc Min Partial DSOP}.
\begin{quote} 
{\sc 3D Matching}

{\sc Input: } A positive integer $n$, a partition $\Pi$ of the set $\{1, 2, \ldots, n\}$ into three sets of equal size, and a collection $\mathcal  S$ of subsets of $\{1, 2, \ldots, n\}$, where every subset contains exactly one element from each of the set of $\Pi$.

{\sc Question: } Is there a subcollection $\mathcal C \subseteq \mathcal S$ of size $n/3$ whose union is $\{1, 2, \ldots, n\}$?
\end{quote}
Note that such a subcollection $\mathcal C$ would provide an exact cover, as it covers each element of the set $\{1, 2, \ldots, n\}$ exactly once.

In the next theorem, we give the reduction from {\sc 3D Matching} to {\sc Min Partial DSOP}. It is basically the same reduction given in~\cite{AHMPS08}  for SOP minimization, however here we show how it works even for disjoint SOP minimization.
Given $u, v \in \{0,1\}^n$, we will write $u \le v$ if $u_i \le v_i$ for all $i \in \{1, 2, \ldots, n\}$.
\begin{theorem}
{\sc Min partial DSOP} is NP-complete.
\end{theorem}
\begin{proof}
We have already noticed that {\sc Min partial DSOP} belongs to NP. Thus, we are left to show that it is NP-hard. To this aim we show how to transform an input instance  of {\sc 3D Matching} into an instance of {\sc Min partial DSOP} in polynomial time. The instance defines an incompletely specified function $f$ depending on $O(\log n)$ variables, that can be covered by a partial DSOP with $n/3$ products if and only if 
there is a subcollection $\mathcal C \subseteq \mathcal S$ of size $n/3$ whose union is $\{1, 2, \ldots, n\}$.

Let $(n, \Pi, \mathcal S)$ be an input instance of {\sc 3D Matching}. We first define two sets of  vectors, $V$ and $W$,  that we will use to define an instance of {\sc Min partial DSOP}. 
Let $q$ be the smallest even integer such that ${q \choose q/2} \ge n$. Observe that $q = O(\log n)$.
We assign a unique q-bit vector $b(i)$ with exactly $q/2$ 1's to each $i \in \{1, 2, \ldots, n\}$. 
Let $\Pi(i) \in \{1,2,3\}$  be the index of the block of the partition $\Pi$ that contains $i$. Let $t = 3q$. The vectors in $V$ and $W$ can be divided into 3 blocks, each of size $q$.
We can now define the vectors in $V = \{v^{(i)}\ |\ 1 \le i \le n \}$ 
and $W = \{w^{(A)}\ | \ A \in \mathcal S \}$:
\begin{itemize} 
\item each $v^{(i)}\in V$, $i \in \{1, 2, \ldots, n\}$, is equal  to $b(i)$ on block $\Pi(i)$, and is 0 in the other two blocks;
\item each $w^{(A)} \in W$, $A \in \mathcal S$, is the bitwise OR of all $v^{(i)} \in V$ such that $i \in A$.
\end{itemize}
These two sets can be generated in time $n^{O(1)}$. 

Observe that this choice guarantees that:
\begin{equation} 
\forall\, A \in \mathcal S, \  \forall\, i \in \{1, 2, \ldots, n\}, \qquad i \in A \ \Longleftrightarrow \ v^{(i)} \le w^{(A)}\,. 
\end{equation}
The forward implication is obvious. To see that the backward implication holds, let $A \in \mathcal S$ and $i \in \{1, 2, \ldots, n\}$, and assume that $v^{(i)} \le w^{(A)}$. This implies that $A$ contains one element $j$ that belongs to the same block $\Pi(i)$ of $i$, where $v^{(i)}$ is not 0, i.e., $\Pi(i) = \Pi(j)$. Thus, since $v^{(i)} \le w^{(A)}$, we must have $b(i) \le b(j)$, which in turn implies $i = j$, and therefore $i \in A$.

We now construct an incompletely specified function $f$ on the domain $\{0,1\}^t$, as follows:
\begin{itemize}
\item $f(x) = 1$ if $x \in V$. 
\item $f(x) = -$ if $x \not\in V$ and $x \le w$ for some $w \in W$.
\item $f(x) = 0$, otherwise.
\end{itemize}
For $u \in \{0,1\}^t$, let $D(u) = \{ w \ |\ w \le u \}$ and let $\tau(u)$ denote  the product $\prod_{i: u_i = 0} \overline x_i$. Observe that $\tau(u)$ is the characteristic function of the set $D(u)$.
Consider the set $D(W) = \bigcup_{x \in W} D(x)$. Property (1) implies that $V \subseteq D(W)$ and that $f(x) = -$ iff $x \in D(W) \setminus V$.

To complete the proof we must show that $\mathcal S$ contains a cover $\mathcal C$  of size $n/3$ if and only if there is a partial DSOP for $f$, with $n/3$ products.
Suppose that $\mathcal S$ contains a cover $\mathcal C$  of size $n/3$, and consider the set of products $\{\tau(w^{(C)}) \ |\ C \in \mathcal C  \}$. It is immediate to verify that the sum of these products covers the function $f$. Indeed, for all $i \in \{1, 2, \ldots, n\}$, $i$ belongs to one of the sets in $\mathcal C$, say $C'$, and the vector $v^{(i)}$ in the on set of $f$ is then covered by the corresponding product  $\tau(w^{(C')})$ (recall that by construction $v^{(i)} \le w^{(C')}$, thus $v^{(i)} \in D(w^{(C')}))$. Now, we have to prove that these products define a partial DSOP for $f$, i.e., we must show that the corresponding cubes are either disjoint or intersect only on the don't cares of $f$. 

First of all recall that each vector $w^{(C)}$, $C \in \mathcal C$, can be divided into three blocks, each equal to one of the vectors $b(i)$. For instance, if $C = \{i, j, k\}$, with $\Pi(i) = 1$, $\Pi(j) =2$, and $\Pi(k) =3$, then $w^{(C)}$ is given by the concatenation of $b(i)$, $b(j)$, and $b(k)$. The related product $\tau(w^{(C)})$ can then be divided into three subterms of $q/2$ literals, containing the  complemented variables corresponding to the 0's in $b(i)$, $b(j)$, and $b(k)$. 
Moreover, since $C$ is a disjoint cover of $\{1, 2, \ldots, n\}$, each $v^{(i)}$ belongs to one and only one of the sets in $C$, that is only one of the vectors $w^{(C)}$ has a block equal to $b(i)$, for all $i \in \{1, 2, \ldots, n\}$. This implies that all subterms of the set of products $\{\tau(w^{(C)}) \ |\ C \in \mathcal C  \}$ are different.

Given any pair of products $\tau(w^{(C)})$ and $\tau(w^{(D)})$, with $C, D \in \mathcal C$, consider the intersection of the corresponding cubes. The characteristic function of the intersection is simply the product (AND) between $\tau(w^{(C)})$ and $\tau(w^{(D)})$. Since all subterms of $\tau(w^{(C)})$ and $\tau(w^{(D)})$ are different, the product $\tau(w^{(C)}) \cdot \tau(w^{(D)})$ contains three subterms, each of at least $q/2 + 1$ complemented variables. Thus, it  can cover only don't cares of $f$, since any vector $v^{(i)} \in V$  has one block with only $q/2$ 0's.

\medskip
Now, suppose that $\phi$ is a partial DSOP for $f$, with $n/3$ products. For each product $p \in \phi$, let $u(p)$ be the maximal vector satisfying $p$. Note that $f(u(p)) \in \{1, -\}$, thus $u(p) \in D(W)$ and there must be a set $S(p) \in \mathcal S$ such that $u(p) \le w^{(S(p))}$. 
We then show that the collection $\mathcal C = \{ S(p) \ |\ p \in \phi \}$ is a cover of $\{1, 2, \ldots, n\}$.  
Let $j \in \{1, 2, \ldots, n\}$. Since $f(v^{(j)}) = 1$, exactly one of the product in $\phi$, say $p^{(j)}$, must cover $v^{(j)}$. This implies $v^{(j)} \le u(p^{(j)})$. Thus $v^{(j)} \le w^{(S(p^{(j)}))}$, which by property (1) implies $j \in S(p^{(j)})$.
\end{proof}

\bigskip


The exponential nature of partial DSOP minimization justifies the search for heuristic solutions. This will be done after a theoretical discussion on how cubes gets fragmented due to their intersections, contained in the next section. This will lead to a heuristic strategy whose complexity is polynomial in the size of the output, i.e., in the number of products of the computed DSOP form.


\section{The Weight of a Cube}
\label{sec-weight}

A product $q=y_{i_1}y_{i_2}...y_{i_k}$, $1 \leq k \leq n$, represents a cube of dimension $d(q)=n-k$, i.e., a cube of 
$2^{n-k}$ points in $\{0,1\}^n$.  The intersection $p=p_1 \cap p_2$ of two cubes $p_1=y_{i_1}...y_{i_{k_1}}$, $p_2=y_{j_1}...y_{i_{k_2}}$ is obviously obtained as the AND of the two corresponding products. The intersection  $p$ is empty if and only if there is a literal in $p_1$ that appears complemented in $p_2$, and vice-versa. Otherwise $p$ is a cube of dimension $d(p)=r$, with $r=n-(k_1 + k_2-c)$, and $c$ is the number of common literals in $p_1$ and $p_2$. 

Take $p_1$, $p_2$ as above, and let $p_1$, $p_2$ partially overlap. The set of points of $p_2\setminus p_1$ can be covered in different ways by a set of at least $k_1-c$ disjoint cubes of dimensions $r,r+1,...,n-k_2-1$. For $n=6$, letting $k_1=5$, $k_2=3$, $c=2$ we have $r=0$ and $d(p_1)=1$, $d(p_2)=3$, i.e., the intersection contains 1 point, and the two cubes contain 2 and 8 points, respectively. Therefore, $p_2\setminus p_1$ contains 7 points and can be covered with $5-2=3$ cubes of dimensions 0, 1, 2.
For an other example, consider cubes $A$ and $B$ in Figure~\ref{Fig:SOP}(a). The set $A \setminus B$ contains the minterms $0000$, $0001$, and $0100$. The disjoint covers for these points are  $\overline x_1 \overline x_2 \overline x_3 + \overline x_1 x_2 \overline x_3 \overline x_4$ and $\overline x_1 \overline x_3 \overline x_4 + \overline x_1 \overline x_2 \overline x_3  x_4$, both  containing two cubes.

 Now, if $p_1$ is selected into a DSOP, $p_2$ must be discarded and the points of $p_2\setminus p_1$ must be covered with at least $k_1-c$ disjoint cubes instead of one (the single $p_2$). Then $k_1-c-1$ is the number of extra cubes required by the DSOP. If the function $f$ can be represented by a SOP containing only $p_1$ and $p_2$, the selection of $p_1$ into a DSOP requires a total of $k_1-c+1$ cubes. In particular if $k_1-c=1$ the intersection $p$ covers exactly one half of the points of $p_2$ and $p_2\setminus p_1$ is also a cube. Clearly the general situation will not be that simple as the starting SOP for $f$, to be transformed into a minimal DSOP, will consist of a collection of cubes overlapping in groups. Still we define a weight for each cube $p_i$ equal to the minimum number of extra cubes that the selection of $p_i$ would induce in all the cubes intersecting $p_i$. 
Formally, let a SOP for $f$ consist of partially overlapping products $p_1,p_2,...,p_s$. We pose:

\begin{definition}
\label{def-weight}
Let a product $p_i$ of $k$ literals intersect the products $p_{i_1},...,p_{i_t}$, such that $p_i$ and $p_{i_j}$ have $c_j$ common literals. Then $w(p_i/p_{i_j})=k-c_j-1$ is the {\em weight of} $p_i$ {\em relative to} $p_{i_j}$, and $w(p_i)=\sum_{j=1}^{t}w(p_i/p_{i_j})$ is the {\em weight} of $p_i$. If $p_i$ does not intersect any other product, set $w(p_i)=-1$.
\end{definition}

Thus, when $p_i$ intersects $p_{i_j}$, the weight of $p_i$ relative to $p_{i_j}$ is the minimum number of additional products that we would have in the cover keeping $p_i$ and covering $p_i/p_{i_j}$ with non-overlapping products.

As an example, consider the function $f$ of four variables, represented in  Figure~\ref{Fig:SOP}(a). A minimal SOP of $f$ contains four cubes $A = \overline x_1\overline x_3$, $B= x_2x_4$, $C= \overline x_1 x_2$, $D= x_1 x_3$, all of dimension two. The weights are computed as follows.
For $A$: $w(A/B)=1$ (in fact, selecting $A$ in a DSOP would require to cover the remaining three points of $B$ with at least two disjoint cubes); $w(A/C)=0$ (the residual two points of $C$ can be covered with one cube); then $w(A)=1$. For $B$: $w(B/A)=1$;  $w(B/C)=0$;  $w(B/D)=1$;  then $w(B)=2$.  For $C$: $w(C/A)=0$;  $w(C/B)=0$;   then $w(C)=0$.  For $D$: $w(D/B)=1$;  then $w(D)=1$. As we shall explain in the next section, we start the construction of a DSOP by selecting the cubes with low weight and high dimension, breaking on the fly the ones that intersect a selected cube. In the present example, start by selecting $C$ and reduce $A$ and $B$ to two subcubes $A_1$, $B_1$ of two points each. Then select $D$ and further reduce $B_1$ to $B_2$ of one point. Then select $A_1$ and $B_2$, as shown in the DSOP of Figure~\ref{Fig:SOP}(b). During the process the weights are updated as explained below.
 
\begin{figure*}[!t]
\begin{center}
\includegraphics[scale=0.30]{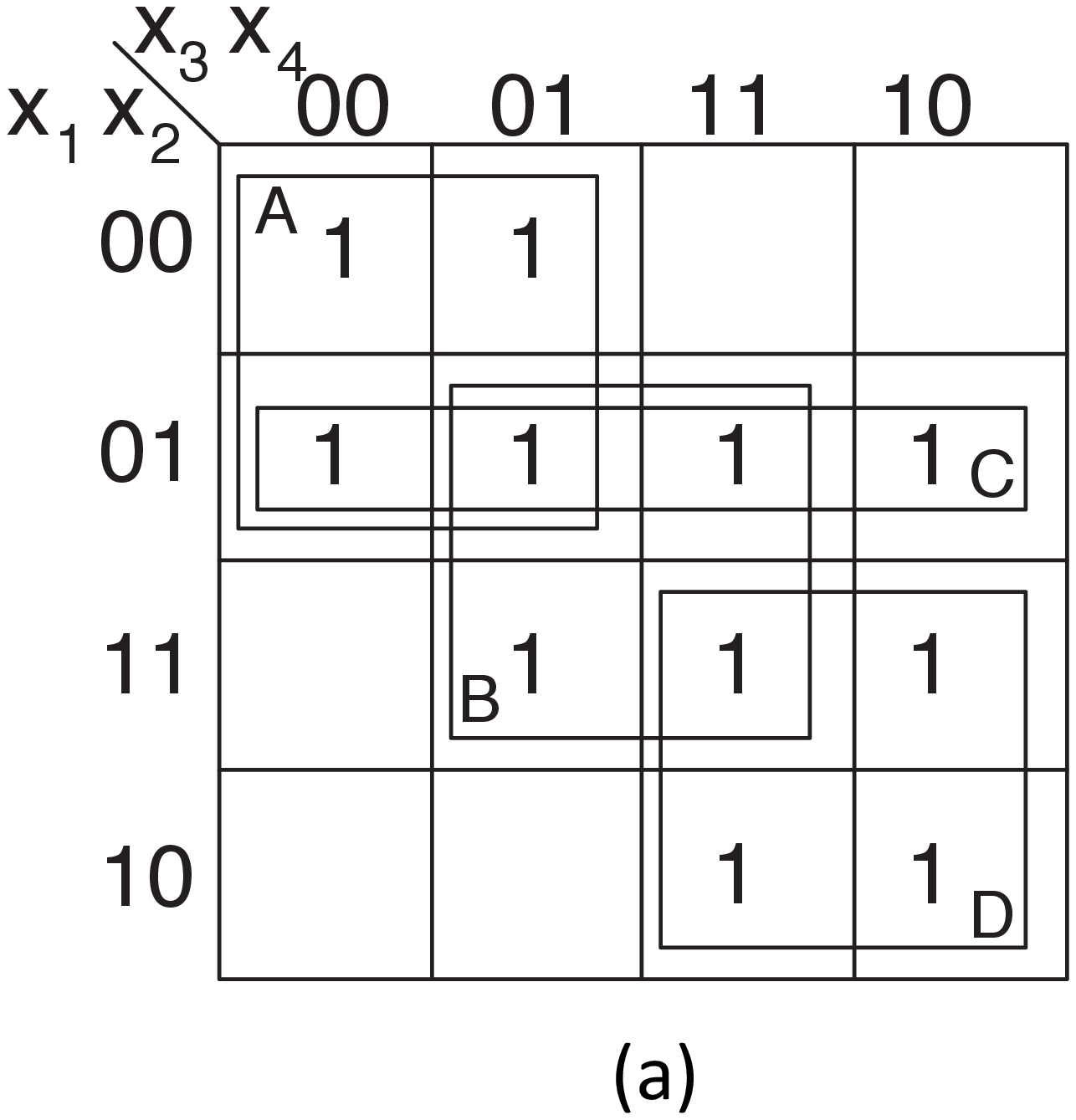} \ \ \ \ \ 
\includegraphics[scale=0.30]{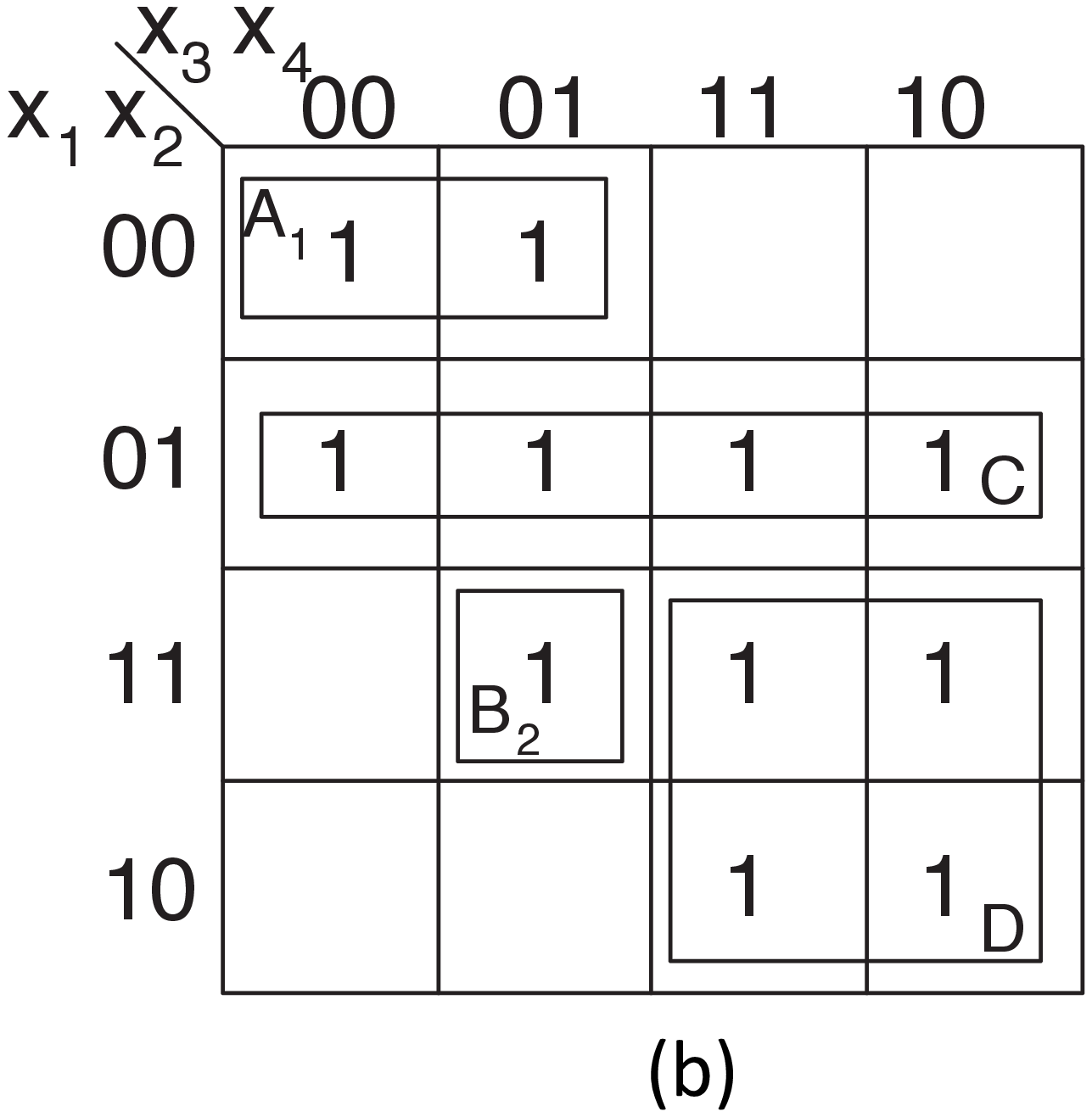}
\vspace{0.0cm}\caption{\label{Fig:SOP} (a) A minimal  SOP of four cubes of dimension 2 in $\cal{B}$$^4$,  with weights $w(A)=1$, $w(B)=2$, $w(C)=0$, $w(D)=1$. (b) A corresponding DSOP.}
\end{center}
\end{figure*}

\section{DSOP synthesis algorithms}
\label{sec-alg}
Let us consider an incompletely defined Boolean function $f: \{0,1\}^n \rightarrow \{0,1,-\}$ represented with a set of cubes $C=(C_{on},C_{dc})$, where $C_{on}$ covers the on set of $f$, i.e., the points $v$ in  $\{0,1\}^n$ such that $f(v)=1$, and $C_{dc}$ covers the don't care set of $f$, i.e., the points $v$ in  $\{0,1\}^n$ such that $f(v)=-$.

The new heuristic for DSOP construction uses four basic procedures working on an explicit representation of cubes. The first procedure BUILD-SOP($C,P$) works on a set $C$ of cubes covering an arbitrary function as above, to build a minimal (or quasi minimal) SOP $P$ for that function. Note that, during the process, BUILD-SOP may be called on different sets $C$ emerging in the computation. As a limit the cubes of $C$ may be minterms, i.e., cubes of dimension 0. 
The second procedure WEIGHT($P$) builds the weights for the cubes of a set $P$. 

The third procedure SORT($P$) sorts a set $P$ of weighted cubes. This procedure comes in two versions: i) the cubes are ordered for decreasing dimension and, if the dimension is the same, for increasing weight; ii) the cubes are ordered for increasing weight and, if the weight is the same, for decreasing dimension. If two or more cubes have same weight and same dimension, their order is chosen arbitrarily. The two versions of SORT give rise to two different alternatives of the overall algorithm.

The fourth procedure BREAK($q,p,Q$) works on the set difference $q\setminus p$ between two cubes, to build an arbitrary minimal set $Q$ of disjoint cubes covering $q\setminus p$. Note that this operation is easy since $q\setminus p$ can be  obtained as $q \setminus (p \cap q)$, where the latter is the set difference between two cubes, i.e., $q$ and $p \cap q$,  in turn a cube because is the intersection of two cubes. 

In practice, for BUILD-SOP one can use any minimization procedure (in our experiments we have used procedure {\sc espresso-non-exact} of the {\sc espresso} suite~\cite{Y91}). Procedures WEIGHT and SORT (both versions) are obvious. Procedure BREAK is the one suggested in~\cite{HCO74} and~\cite{Sa99}  as DISJOINT-SHARP. 

In the overall process we consider four sets of cubes $C,P,B,D$. At the beginning $C$ contains the cubes defining $f$, while $P,B,D$ are empty. During the process $C$ contains the cubes defining the part of $f$ still to be covered with a DSOP; $P$ contains the cubes of a SOP under processing; $B$ temporarily contains cubes produced by BREAK as fragmentation of cubes of $P$; and $D$ contains the cubes already assigned to the DSOP solution and, at the end, the solution itself. 

\begin{figure*}
\centering
\begin{minipage}{.5\textwidth}
\small
\begin{tabbing}
ddd\=ddd\=ddd\=ddd3=\kill

{\bf algorithm} DSOP($C,D$) \\ 
{\bf INPUT: } A set of cubes $C$ covering a function $f$\\
{\bf OUTPUT: }  A set of disjoint cubes $D$ covering $f$\\ \\
$D=\emptyset$ \\
{\bf while} ($C \neq \emptyset$)\\
\> BUILD-SOP($C, P$)\\
\> $A = \{d\in P \ |\  \forall c \in P\setminus\{d\}:\ d \cap c = \emptyset \}$\\ 
\> $D = D \cup A$\\
\> $P = P\setminus A$ \\ 
\> WEIGHT($P$) \\
\> SORT($P$) \\
\>  $B=\emptyset$ \\
\> {\bf while} ($P \neq \emptyset$) \\
\>\> {\bf let} $p$ be the first element of $P$ \\
\>\> $P = P\setminus \{p\}$ \\ 
\>\> $D = D \cup \{p\}$ \\ 
\>\> {\bf forall} $q\in P:\, p\cap q \neq \emptyset$ \\
\>\>\> $P = P\setminus \{q\}$  \\
\>\>\> BREAK($q, p,Q$)\\
\>\>\> OPT($q,Q,P,B$) \\
\>\> {\bf forall} $r\in B:\, p\cap r \neq \emptyset$  \\
\>\>\>  $B = B\setminus \{r\}$ \\
\>\>\>  BREAK($r, p,Q)$\\
\>\>\>  $B=B\cup Q$ \\
\> $C=B$ \\

\end{tabbing}
\end{minipage}
\caption{\label{Fig:alg} The general algorithm for DSOP synthesis.}
\end{figure*}

The algorithms of our family share the structure shown in Figure~\ref{Fig:alg} (its behaviour on incompletely specified functions is discussed at the end of this section).
As long as $f$ has not been completely covered with disjoint cubes, i.e., there are still cubes in the set $C$, a minimal (or quasi-minimal) SOP $P$ for the part of $f$ still to be covered is computed by the procedure BUILD-SOP. All cubes that do not intersect any other cube in $P$ are removed from $P$ and inserted in the DSOP $D$ under construction; the remaining cubes are weighted and sorted.
Then, the first cube $p$ is extracted from $P$ and inserted in the solution $D$. Each cube  $q \in P$ that intersects $p$ is removed from $P$, and a SOP $Q$ for the set difference $q \setminus p$ is computed by the procedure BREAK. 

During this phase an optional optimization procedure OPT is called to decide how to handle the fragments in $Q$; depending on this  optimization phase, different variants of the heuristic can be defined.
Note that, since the points of $p$ cannot be covered by any other cube, all fragments $r$ already inserted in $B$  must be tested for intersection with $p$ and, if necessary, replaced with the SOP computed by BREAK for the set difference $r \setminus p$.
When $P$ becomes empty, the fragments in $B$ are moved to the set $C$ and the algorithm iteratively builds a new SOP $P$ covering the points that are not yet covered by the DSOP $D$ under construction. The iterations terminate when $C$ becomes empty.

\bigskip

We have designed and tested five variants of our heuristic based on five different versions of the optimization procedure OPT, with different degrees of sophistication.  
The first variant, {\bf DSOP-1}, is the simplest, and computationally fastest, as OPT simply inserts the cubes of $Q$ into the set of fragments $B$:

\medskip
\begin{quote}
\begin{tabbing}
ddd\=ddd\=ddd\=ddd3=\kill
in DSOP-1: \\
\> {\bf procedure} OPT($q,Q,P,B$)\\
\> \> $B = B\cup Q$ \\
\end{tabbing}
\end{quote}

\begin{example}\label{ex:DSOP}
For an example, Figure~\ref{Fig:SOP}(b) shows a DSOP form for the SOP form of Figure~\ref{Fig:SOP}(a), computed by algorithm DSOP-1.
At the beginning $D = \emptyset$ and $P = \{\overline x_1 x_2,  x_1 x_3 , \overline x_1 \overline x_3, x_2  x_4 \}$, sorted for decreasing dimensions of cubes and then for increasing weights (we recall that,  $w(\overline x_1 x_2) = 0$, $w(x_1 x_3) = 1$, $w(\overline x_1 \overline x_3) = 1$, and $w(x_2 x_4) = 2$). The first cube considered is  $p=\overline x_1 x_2$, which is removed from $P$ and inserted in $D$. Its intersecting cubes, $\overline x_1 \overline x_3$ and $x_2 x_4$, are then broken generating the residuals cubes $\overline x_1 \overline x_2 \overline x_3$ and $x_1 x_2 x_4$, respectively, which are inserted in $B$, while $\overline x_1 \overline x_3$ and $x_2 x_4$ are removed from $P$. The last cube in $P$ to be considered is then $x_1 x_3$ that is inserted directly in $D$, since there are not any other remaining cubes in $P$. Its intersecting cube $x_1 x_2  x_4$ in $B$ is then reduced to $x_1 x_2 \overline x_3 x_4$. 
The second {\bf while} $(P\neq \emptyset)$ iteration starts with $P=\{\overline x_1 \overline x_2 \overline x_3, x_1 x_2 \overline x_3 x_4\}$ and $D = \{\overline x_1 x_2, x_1 x_3\}$, and terminates with the final DSOP $D = \{\overline x_1 x_2, x_1 x_3, \overline x_1 \overline x_2 \overline x_3, x_1 x_2 \overline x_3 x_4\}$.
\end{example}

\bigskip

In the second variant, {\bf DSOP-2}, after  a cube $p$ has been selected and moved to $D$, 
each cube $q$ intersecting $p$ is, as before, fragmented and moved to $B$. In addition the optimization procedure updates the weight of all cubes $r \in P$ that intersect $q$, and then sorts the cubes in $P$ again:

\medskip
\begin{quote}
\begin{tabbing}
ddd\=ddd\=ddd\=ddd3=\kill
in DSOP-2:\\
\> {\bf procedure} OPT($q,Q,P,B$)\\
\> \> $B = B\cup Q$ \\
\> \> $I = \{ r\in P \ | \ q \cap r \neq \emptyset  \}$ \\
\> \> WEIGHT($I$) \\
\> \> SORT($P$) \\
\end{tabbing}
\end{quote}

A disadvantage of both versions is that whenever a cube $p$ is moved from $P$ to $D$, all cubes $q$ intersecting $p$ are fragmented and removed from the set $P$. Hence, the fragments, even the big ones,  are ``out of the game'' and cannot participate in the construction of the DSOP $D$ until $P$ becomes empty and a new SOP covering all fragments in the set $B$ is computed. Consequentially, small cubes in $P$ could be selected first, possibly damaging the quality of the final result, i.e., the size of the final DSOP.

\bigskip

To partially avoid this disadvantage, we have implemented a third version of the heuristic, {\bf DSOP-3}, in which whenever a cube $p \in P$ is moved to $D$, each cube $q$ intersecting $p$ is, as before, fragmented and moved to $B$, and, in addition, all cubes $r \in P$ intersecting $q$ are moved to $B$ as well:

\medskip
\begin{quote}
\begin{tabbing}
ddd\=ddd\=ddd\=ddd3=\kill
in DSOP-3:\\
\> {\bf procedure} OPT($q,Q,P,B$)\\
\> \> $I = \{ r\in P \ | \ q \cap r \neq \emptyset  \}$ \\
\> \> $B = B\cup Q \cup I$ \\ 
\> \> $P = P \setminus I$ \\
\end{tabbing}
\end{quote}
In this way, the cubes of $P$ intersecting the fragments already in B cannot be selected,while is avoided the possible fragmentation of big cubes in $B$. Moreover, we leave open the possibility of selecting these big cubes in the next iterations of the algorithm, 
This version of the heuristic is computationally more expensive, since  in the internal while loop less cubes can be selected ($P$ empties faster), and procedure BUILD-SOP must be executed more frequently. 

\bigskip

The fourth version of the heuristic, {\bf DSOP-4},  checks whether the set $Q$ contains only one fragment, i.e., $q \setminus p$ is a cube. In this case, this only fragment is put back in $P$.  The cubes left in $P$ are then weighted and sorted again:
\medskip
\begin{quote}
\begin{tabbing}
ddd\=ddd\=ddd\=ddd3=\kill
in DSOP-4:\\
\> {\bf procedure} OPT($q,Q,P,B$)\\
\> \> {\bf if} ($|Q| = 1$) \\ 
\> \> \> $P = P \cup Q$ \\ 
\> \> {\bf else}  \\
\> \> \> $B = B\cup Q$ \\ 
\> \> WEIGHT($P$) \\
\> \> SORT($P$) \\
\end{tabbing}
\end{quote}

\bigskip

Finally, in the last version of the heuristic that we have tested, {\bf DSOP-5}, the biggest fragment in the set $Q$ is always put back in $P$. The cubes left in $P$ are then weighted and sorted again.
In this way,  big fragments remain part of the game in the present iteration of the algorithm:
\medskip
\begin{quote}
\begin{tabbing}
ddd\=ddd\=ddd\=ddd3=\kill
in DSOP-5:\\
\> {\bf procedure} OPT($q,Q,P,B$)\\
\> \> {\bf let} $b$ be the biggest cube in $Q$ \\
\> \> $P = P \cup \{b\}$ \\ 
\> \> $B = B\cup Q \setminus\{b\}$ \\ 
\> \> WEIGHT($P$) \\
\> \> SORT($P$) \\
\end{tabbing}
\end{quote}

The performances of these five procedures are discussed in Section~\ref{sec-exp}.
We have observed experimentally that  more sophisticated optimization procedures do not always provide better quality results. Experimental results have also outlined how the BUILD-SOP procedure, i.e., re-synthesizing the remaining cubes, seems to be crucial for obtaining compact DSOPs.

Let us now briefly consider the case of the DSOP synthesis of incompletely specified Boolean functions. Our heuristic does not consider explicitly the presence of don't cares; indeed, the first call of the BUILD-SOP procedure produces a SOP $P$ covering the whole on set of $f$ and a subset of its don't care set. Then, the algorithm works on the SOP $P$, treating all points covered by its cubes as if they belonged to the on set of $f$, i.e., there is no distinction between points originally in the on set of $f$ and points originally in the don't care set. In particular, the successive calls of BUILD-SOP on the part of $f$ still to be covered with a DSOP, treat the function as if it  were completely specified.
Of course, each cube in the SOP $P$ computed by the first call of BUILD-SOP covers at least one point in the on set of $f$, as cubes covering only points in the don't care set  are discarded by the SOP minimization algorithm. However, the final disjoint cover $D$ for $f$ could contain cubes covering only points originally in the don't care set. In fact, cubes in $D$ are either entire cubes of  the starting SOP $P$,  or sub-cubes of cubes in $P$ (besides new cubes and sub-cubes originated by the successive calls of BUILD-SOP) and some sub-cubes (or new cubes) could only cover don't care points.

From the above all versions of our heuristic could be improved checking whether a cube $p$ contains only points in the don't care set of the function $f$, before adding it to the DSOP solution $D$ under construction. Unfortunately, such a check can be computationally expensive, and for this reason we have not added it as a ``default'' procedure in our algorithm.
In fact the check is left as an option. Experiments conducted on a set of incompletely specified functions show some improvements on the final form induced by the check at a considerable increase of computing time, see next Section~\ref{sec-exp}.


\section{Partial DSOP synthesis}
\label{sec-alg2}
As already mentioned the problem of DSOP minimization naturally generalizes to covering partial DSOPs where some minterms (e.g. the ones in the on set of the function) are covered exactly once while other minterms (e.g. the ones in the don't care set) can be covered any number of times~\cite{LF09}.
In this section we present a general heuristic to efficiently compute  a {\em partial DSOP} cover.

The heuristic  makes use of two sums of products as input.
The first SOP, $sopD$, contains all points of the on and don't care set of the function $f$ that must be covered only once (DSOP part), while the second SOP, $sopS$, contains all the points of $f$ that can be covered more than once  (SOP part). These two SOPs are disjoint.
The output of the heuristic is a cover of the overall function $f$, represented by the union of the two SOPs $sopD$ and $sopS$ that respects the  specifications. Note that when $sopD$ is empty the problem is a classical SOP minimization, while when $sopS$ is empty the problem is a classical DSOP minimization.

The algorithm uses four basic procedures as for the DSOP synthesis of Section~\ref{sec-alg}. In particular BUILD-SOP, WEIGHT, and SORT are the same.  

The fourth procedure PARTIAL-BREAK($q, p, sopD, sopS, Q, R$) works on the set difference $q\setminus p$ between two cubes, to build an arbitrary minimal set $Q$ of disjoint cubes covering $q\setminus p$, if $q\cap p$ is not entirely contained in $sopS$. If $q\cap p$ is contained in $sopS$, the cube $q$ is not broken and we can keep it in the set $P$ which contains the cubes to be considered in the current iteration. 
In this case we then set  $Q= \emptyset$.
Moreover, the procedure PARTIAL-BREAK builds a set $R$ containing points of $q\setminus p$ that can be covered more than once and can therefore be added as don't cares to $C$. In this way, these points, that have been already covered, could be used again in the minimization phase to get a smaller cover. This procedure, different from the one used for DSOP synthesis,  is presented in Figure~\ref{fig:break}.

\begin{figure*}
\centering
\begin{minipage}{.5\textwidth}
\small
\begin{tabbing}
ddd\=ddd\=ddd\=ddd3=\kill

{\bf algorithm} PARTIAL-BREAK($q, p, sopD, sopS, Q, R$)  \\
{\bf INPUT: } The chosen cube $p$, the cube $q$ that can be broken, \\the two SOPs $sopD$ and $sopS$ whose union represents $f$\\
{\bf OUTPUT: } A minimal set $Q$ of disjoint cubes covering $q\setminus p$ and\\ a set $R$ of the points of $q\setminus p$ that can be covered more than once\\ \\

$R = \emptyset$ \\
$p_i = q \cap p$ \\
{\bf if } ($p_i \subseteq sopS$) {\em // all points of $p_i$ can be covered more than once}\\
\> $Q = \emptyset$ \\
{\bf  else if} ($p_i \subseteq sopD$) {\em // all points of $p_i$ must be covered once}\\
\> $Q =$ DISJOINT\_SHARP($q,p_i$) \\
{\bf  else} {\em // $p_i$ intersects both $sopD$ and $sopS$}\\
\> $Q =$ DISJOINT\_SHARP($q,p_i$)\\
\> $R = p_i \cap sopS$ \\
\end{tabbing}
\end{minipage}
\caption{\label{fig:break} The procedure PARTIAL-BREAK to be used in partial DSOP synthesis.}
\end{figure*}

\begin{figure*}
\centering
\begin{minipage}{.5\textwidth}
\small
\begin{tabbing}
ddd\=ddd\=ddd\=ddd3=\kill

{\bf algorithm} PARTIAL-DSOP($sopD, sopS, D$)  \\
{\bf INPUT: } Two disjoint SOPs describing the points of $f$ that \\must be covered only once ($sopD$) and the points of $f$ that \\can be covered more than once  ($sopS$)\\
{\bf OUTPUT: } A partial DSOP $D$ for the function $f$\\ \\ 

$C_{on} = sopD_{on} \cup sopS_{on}$ \\
$C_{dc} = sopD_{dc} \cup sopS_{dc}$ \\

{\bf while} ($C_{on} \neq \emptyset$) \\

\> BUILD-SOP($C, P$)\\
\> $A = \{d\in P \ |\  \forall c \in P\setminus\{d\}:\ d \cap c = \emptyset \}$\\ 
\> $D = D \cup A$\\
\> $P = P\setminus A$ \\ 
\> WEIGHT($P$) \\
\> SORT($P$) \\
\> $B=\emptyset$ \\
\> {\bf while} ($P \neq \emptyset$) \\
\>\> {\bf let} $p$ be the first element of $P$ \\
\>\> $P = P\setminus \{p\}$ \\ 
\>\> $D = D \cup \{p\}$ \\ 
\>\> {\bf forall} $q\in P:\, p\cap q \neq \emptyset$ \\
\>\>\> PARTIAL-BREAK($q, p, sopD, sopS, Q, R$)\\
\>\>\> {\bf if} ($Q \neq \emptyset$) $P = P\setminus \{q\}$\\
\>\>\> OPT($q,Q,P,B$) \\
\>\>\> $C_{dc} =C_{dc} \cup  R$\\
\>\> {\bf forall} $r\in B:\, p\cap r \neq \emptyset$  \\
\>\>\>  PARTIAL-BREAK($r, p, sopD, sopS, Q, R$)\\
\>\>\> {\bf if} ($Q \neq \emptyset$) $B = B\setminus \{r\}$\\
\>\>\>  $B=B\cup Q$ \\
\>\>\> $C_{dc} =C_{dc} \cup  R$\\
\> $C_{on}=B$ \\
\end{tabbing}
\end{minipage}
\caption{\label{fig-partialDSOP}  Algorithm for partial DSOP synthesis.}
\end{figure*}

\begin{figure*}[t]
\begin{center}
\includegraphics[scale=0.3]{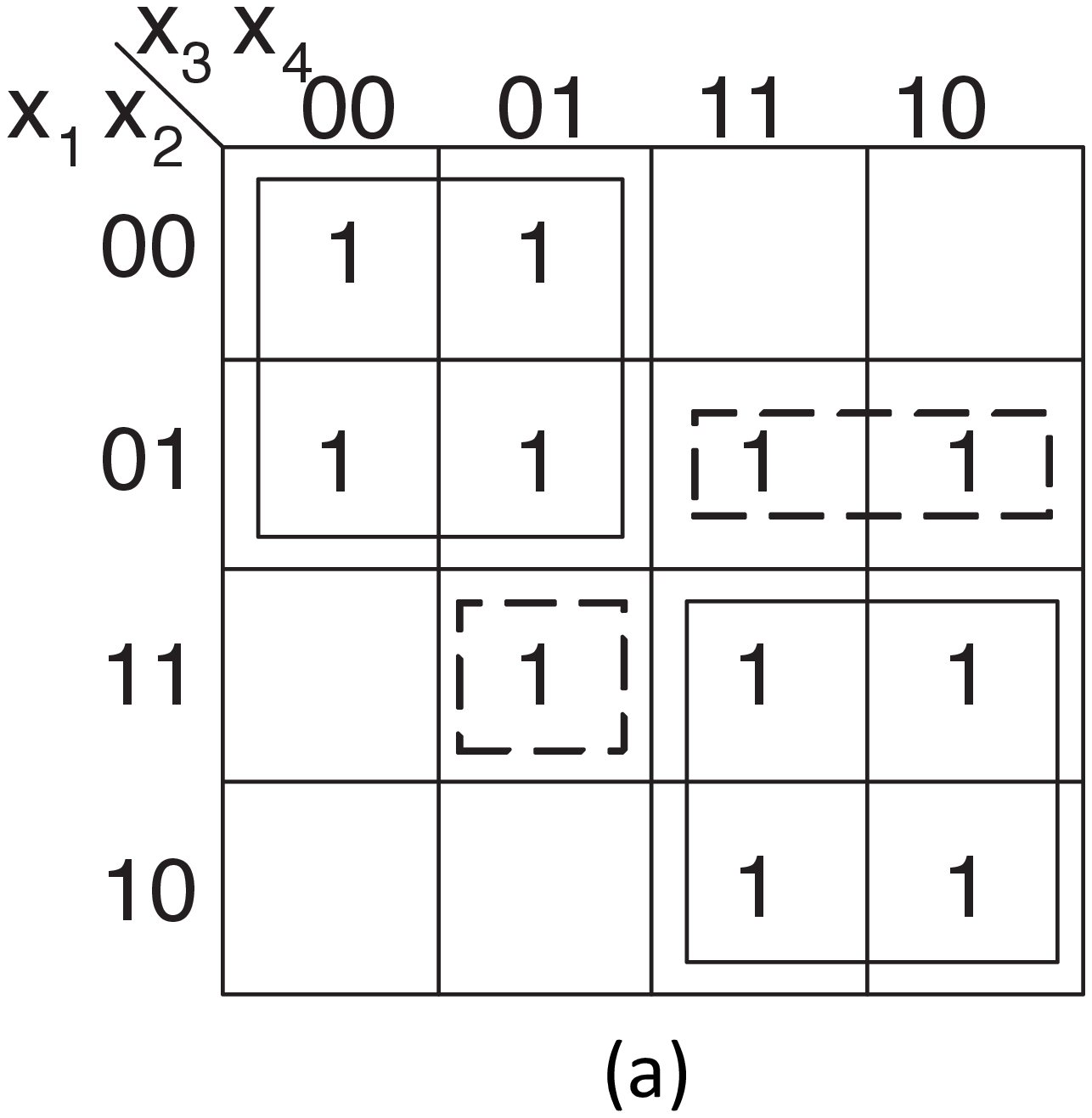}
\hspace{0.8cm}
\includegraphics[scale=0.3]{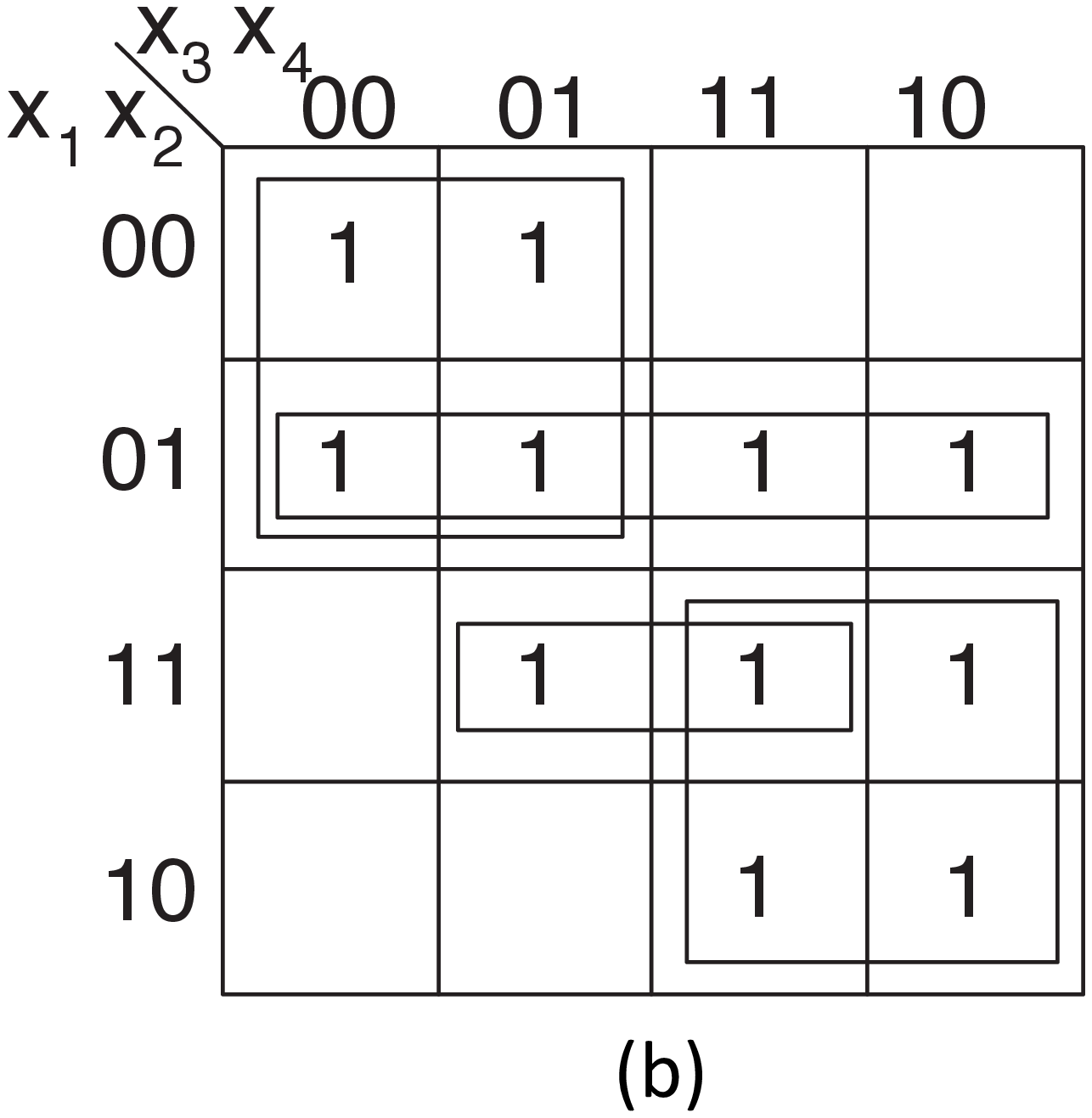}
\vspace{0.0cm}\caption{\label{Fig:SOPDS} (a) sopS (cubes with solid lines) and sopD (cubes with dotted lines). (b) A corresponding partial DSOP.}
\end{center}
\end{figure*}

The overall minimization heuristic is presented in Figure~\ref{fig-partialDSOP}.
As for the DSOP synthesis, the heuristic makes use of four sets of cubes $C,P,B,D$.
At the beginning $C = sopD \cup sopS$ contains the cubes defining $f$ while $P,B,D$ are empty. During the processing $C$ contains the cubes defining the part of $f$ still to be covered with a partial DSOP. $P$ contains the cubes of a SOP under processing. $B$ temporarily contains cubes produced by BREAK as fragmentation of cubes of $P$. $D$ contains the cubes already assigned to the partial DSOP solution and, at the end, the solution itself. 
OPT($q,Q,P,B$) is an optional optimization procedure to decide how to handle the fragments produced by the procedure BREAK. As before, depending on this  optimization phase, different variants of the heuristic can be defined.

\begin{example}\label{ex:PDSOP}
Consider the function shown in Figure~\ref{Fig:SOPDS}(a). Suppose that $sopD = \{\overline x_1 x_2 x_3, x_1 x_2 \overline x_3 x_4\}$ (cubes with dotted lines in the figure) and $sopS = \{\overline x_1 \overline x_3, x_1 x_3\}$ (cubes with solid lines). A partial DSOP for $f$ is shown in Figure~\ref{Fig:SOPDS}(b). This expression is obtained with the partial DSOP algorithm as described in the following.
Let OPT($q,Q,P,B$) be the simple command $B = B\cup Q$ (as in the {\bf DSOP-1} procedure). At the beginning $D = \emptyset$ and, after the SOP minimization phase, $P = \{\overline x_1 x_2,  x_1 x_3 , \overline x_1 \overline x_3, x_2  x_4 \}$, sorted  for decreasing dimensions of cubes and then for increasing weights (note that we have the same initial $P$ of Example~\ref{ex:DSOP}). The first cube  $p=\overline x_1 x_2$  is removed from $P$ and inserted in $D$. 
Its intersecting cubes are $\overline x_1 \overline x_3$ and $x_2 x_4$. In the procedure PARTIAL-BREAK, the intersection between $x_2 x_4$ and $\overline x_1 x_2$ is $p_i = \overline x_1 x_2 x_4$. Note that  $p_i$ intersects both $sopD$ and $sopS$, thus $Q=\{x_1 x_2 x_4\}$ and $R=\{\overline x_1 x_2 \overline x_3 x_4\}$ (i.e., $x_1 x_2 x_4$ and $\overline x_1 x_2 \overline x_3 x_4$ will be inserted in $B$ and in the don't care set of $C$, respectively). Moreover, we compute the intersection $p_i$ between $\overline x_1 \overline x_3$ and $\overline x_1 x_2$, obtaining $\overline x_1 x_2\overline x_3$ which is entirely contained in $sopS$. Thus, in this case $Q=R=\emptyset$, then $\overline x_1 \overline x_3$ is not broken and  it is not removed from $P$. 
Similar operations are performed on $\overline x_1 \overline x_3$, and on $x_1 x_2 x_4$ contained in $B$.
The second {\bf while} $(P\neq \emptyset)$ iteration, which starts with the $P=\{x_1 x_2 x_4\}$ and $D = \{\overline x_1 x_2, x_1 x_3,\overline x_1 \overline x_3\}$, terminates with the partial DSOP shown in Figure~\ref{Fig:SOPDS}(b).
\end{example}

\section{Experimental Results}\label{sec-exp}
In this section we present and discuss the results obtained with the heuristics presented above  to  the standard {\sc espresso} benchmark suite~\cite{Y91}.
All experiments were performed on a 1.8 GHz  PowerPC with 1 GB of RAM.

\begin{table*}[ht]
  \centering
  \begin{scriptsize}
    \begin{tabular}
{||l||r|r||r||r|r||r|r||r|r||r|r||r|r||}
\hline
\multicolumn{14}{||c||}{\bf SORT VERSION: dimension/weight}\\
\hline
 &   &  &{\bf SOP} & \multicolumn{ 2}{c||}{{\bf DSOP-1}} & \multicolumn{ 2}{c||}{{\bf DSOP-2}} & \multicolumn{ 2}{c||}{{\bf DSOP-3}} & \multicolumn{ 2}{c||}{{\bf DSOP-4}} & \multicolumn{ 2}{c||}{{\bf DSOP-5}} \\
 {\bf Bench} &  {\bf in} &{\bf out} &    {\bf size} &  {\bf size}&  {\bf time}&   {\bf size}&   {\bf time}&  {\bf size}& {\bf  time} &{\bf size}& {\bf time} &{\bf size}& {\bf  time} \\\hline

accpla  & 50&69 & 175 & 1457 &11.68 &1458 &13.42& 1190 &10.55&{\bf 1125}&5.61&1528& 21.60\\\hline    
addm4   & 9&8 & 200 &  218 &0.26&221 &0.31&{\bf  214}&0.40&222 &0.19&224&0.19\\\hline
alu4 & 14& 8& 575& 923 &1.51 &921	 &2.00 &	{\bf 881}	&2.32 &1051 &3.36 &	1044 &2.87 \\\hline
apex3& 54&50 & 280& {\bf 345}	&0.62 &	{\bf 345}	&0.65 &	350	&0.68 &366	&2.04 &	400	&2.11 \\\hline
apex4& 9&19 & 436& 506	&0.34 &		506	&0.36 &	503	&0.31 &	502	&0.52 &		{\bf 501}	&0.59 \\\hline
b2      & 16&17 & 106 &  131 &0.42 &  131 &0.49 & \bf 121 &0.54  &130	&0.93 &	127	&0.56 \\\hline
bc0     & 26&11 & 179 &  214 &0.47 &  214 &0.53 &  {\bf  202} &0.68 &  212 &0.53 &		208	&0.51 \\\hline
chkn& 29& 7& 140&187	&0.87 &		187	&0.88 &	{\bf 168}	&0.88 &		215	&0.48 &		221	&0.42 \\\hline
clip& 9&5 & 120& 151	&0.28 &	150	&0.29 &	{\bf 140} &0.39 &	153	&0.20   &157	&0.16 \\\hline
cps&  24&109 &163 &{\bf 184}	&0.75 &		{\bf 184}	&0.76 &		204	&0.89 &	219	&0.38 &		225	&0.38 \\\hline
dist    &  8&5 & 123 &  135 &0.22 &   135 &0.23 &  130 &0.38 & {\bf 128}	 &0.15 &	129	&0.16 \\\hline
ex5     &  8&63 &  74 &  126 &0.79 &  126 &0.44 &{\bf  122} &0.80 &  137 &0.39 &	141	&0.48 	\\\hline
gary    & 15&11 & 107 &  134 &0.52 &    134 &0.35 &{\bf  124} &0.49 &   126 &0.28 &	127	&0.16 	\\\hline
ibm     & 48&17 & 173 &  366 &1.23 &  366 &0.59 &{\bf  361} &0.99 & 373	 &0.46  &		391	&0.30 \\\hline
in4     & 32&20 & 212 &  312 &1.36 &  312 &0.84 &{\bf  280} &1.33 &  303	&0.54 &	304	&0.56 \\\hline
intb    &  15&7 & 631 &  811 &2.03 &    818 &1.61 &{\bf  798} &2.57 &922 &2.56 &	952	 &2.91 \\\hline
jbp     & 36&57 & 122 &  135 &0.57 &  135 &0.26 &{\bf  127} &0.43 &134	 &0.20 &	136	&0.17 \\\hline
mainpla & 27&54 & 172 &  296 &4.67 &   296 &3.00 &  293 &3.23 & 288	 &5.37 &	{\bf 260}	&5.20 \\\hline
max1024&10 &6 &274 &\bf 332	&0.32 &	334	&0.33 &	334	&0.54 &	347	&0.35 &	345	&0.32 \\\hline
misex3  & 14&14 & 690 & 1070 &2.72 & 1073 &1.49 &{\bf 1032} &2.68 &1159	&2.57 &	1309 &3.48 \\\hline
soar    & 83&94 & 353 &  447 &1.93 &    447 &1.30 &  \bf 434 &1.58  &442	&0.59  &	456	&0.58 \\\hline
sym10& 10&1 &210 &232	&0.43 &{\bf 231} &0.51 &		232	&1.11 &		235	&1.01 &		245	&0.85 \\\hline
table3  & 14&14 & 175 &  181 &0.41 &    181 &0.23 & 180 &0.33 & {\bf 179}	 &0.16 &		{\bf 179}	&0.18 \\\hline
table5  & 17&15 & 158 &  167 &0.39 &   167 &0.36 & {\bf  161} &0.38 &{\bf 161}	&0.24 &		{\bf 161}	&0.25 \\\hline
tial    & 14&8 & 581 &  943 &1.78 & 937 &1.96 &{\bf  874} &2.98 &  1071	&2.84 &		1040	 &2.50 \\\hline
vtx1	 & 27& 6&110 &{\bf 204}	&0.45 &	{\bf 204}	&0.49 &	{\bf 204}	&0.64 &	208	&0.34 &		213	&0.31 \\\hline
x7dn	 &66 &15 & 538&796	&1.30 &		{\bf 784}	&1.43 &		812	&1.57 &		813	&0.88 &	864	&0.76  \\\hline
\end{tabular}  
 \end{scriptsize}  
  \caption{Comparison of five different variants of the DSOP minimization heuristic (SORT version: {\bf dimension/weight}.) The size of the best DSOP representation computed for each benchmark is in boldface.}
  \label{tab1dw}
  \end{table*}

\begin{table*}[ht]
 \centering
  \begin{scriptsize}
    \begin{tabular}
{||l||r|r||r||r|r||r|r||r|r||r|r||r|r||}
\hline
\multicolumn{14}{||c||}{\bf SORT VERSION: weight/dimension}\\
\hline
 &   &  &{\bf SOP} & \multicolumn{ 2}{c||}{{\bf DSOP-1}} & \multicolumn{ 2}{c||}{{\bf DSOP-2}} & \multicolumn{ 2}{c||}{{\bf DSOP-3}} & \multicolumn{ 2}{c||}{{\bf DSOP-4}} & \multicolumn{ 2}{c||}{{\bf DSOP-5}} \\
{\bf Bench} &  {\bf in} &{\bf out} &    {\bf size} &  {\bf size}&  {\bf time}&   {\bf size}&   {\bf time}&  {\bf size}& {\bf  time} &{\bf size}& {\bf time} &{\bf size}& {\bf  time} \\\hline
accpla  & 50&69 & 175 & 1779  &24.57 &1717  &32.46 &\bf 1317  &16.80  &1535  &34.29  &3078  &97.73 \\\hline    
addm4   & 9&8 & 200 &   220  &0.26 &   222   &0.27 & \bf 217  &0.29 & 222  &0.20 &223  &0.13 \\\hline
alu4 & 14& 8& 575& 1138  &2.25 &	\bf 1065  &2.12 &	1276  &4.74 &	1269  &5.23 &	1211	  &3.69 \\\hline
apex3& 54&50 & 280& {\bf 337}	 &0.54 &		342	 &0.57 &	347	 &0.6	 &	356	 &0.82 &	386	 &0.94 \\\hline
apex4& 9&19 & 436& 	506	 &0.34 &	506	 &0.36 &	503	 &0.33 &	502	 &0.38 &		{\bf 500}	 &0.25 \\\hline
b2      & 16&17 & 106 &  131  &0.41 &  131  &0.44 &{\bf 120}  &0.62  &	126	 &0.36 &		125	 &0.36 \\\hline
bc0     & 26&11 & 179 &   230  &0.60 &  218  &0.53 &  \bf  210  &1.18  &	218	 &0.44 &211	 &0.40 \\\hline
chkn& 29& 7& 140&	598	 &3.67 &	448	 &3.55 &	\bf	216	 &2.86 &386	 &2.21 &		426	 &1.70 \\\hline
clip& 9&5 & 120& 154	 &0.32 &	153	 &0.29 &	\bf	143	 &0.38 &		157	 &0.2	0  &	159	 &0.17 \\\hline
cps&  24&109 &163 &	{\bf 184}	 &0.75 &		{\bf 184}	 &0.88 &		204	 &1.19 &		217	 &0.36 &	223	 &0.36 \\\hline
dist    &  8&5 & 123 &  138  &0.21 &   138   &0.23 &  \bf 133  &0.36  &	134	 &0.15 &	\bf	133	 &0.15 \\\hline
ex5     &  8&63 &  74 &   128  &0.38 & \bf 125  &0.44 &    142  &0.77  &	141	 &0.39 &		148	 &0.31 \\\hline
gary    & 15&11 & 107 &    139  &0.28 & 131  &0.27 &    132  &0.43  &	{\bf 124}	 &0.15 &	125	 &0.15 \\\hline
ibm     & 48&17 & 173 &    431  &0.69 & \bf 393   &0.64 &    416  &1.24  &	415	 &0.56 &		478	 &0.47 \\\hline
in4     & 32&20 & 212 &   329  &0.72 & 331  &0.81 &    321  &1.29  &\bf 303	 &0.52 &		319	 &0.48 \\\hline
intb    &  15&7 & 631 &   955  &1.79 & \bf 932  &1.83 &   1125  &3.65 &	1130  &4.99 	&	1173	  &3.84 \\\hline
jbp     & 36&57 & 122 &  151  &0.35 & 147   &0.36 &\bf  128  &0.33  &140	 &0.16 &		147	 &0.17 \\\hline
mainpla & 27&54 & 172 &    459  &2.86 & 405  &2.47 &     387  &3.33  &	366	 &2.76 &	\bf	338	 &2.23 \\\hline
max1024&10 &6 &274 &	334	 &0.30 &	330	 &0.36 &	{\bf 324}	 &0.50 &		339	 &0.34 &		338	 &0.29 \\\hline
misex3  & 14&14 & 690 & \bf  1132  &1.28 &1155   &1.75 &  1317  &4.58 &	1234	  &3.67 &		1464  &4.36 \\\hline
soar    & 83&94 & 353 &    451  &1.16 &   449  &1.25 & {\bf 430}  &1.41   &	440	 &0.61  &	464	 &0.60 \\\hline
sym10& 10&1 &210 &\bf 233	 &0.42 &		234	 &0.48 &		248	 &1.37 &		239	 &1.19 &		258	 &1.38 \\\hline
table3  & 14&14 & 175 &   181  &0.21 & 181   &0.24 &  180  &0.24 &{\bf 179}	 &0.14 &		{\bf 179}	 &0.14 \\\hline
table5  & 17&15 & 158 &    167  &0.31 &   167  &0.32 &{\bf 161}  &0.28 &	{\bf 161}	 &0.20 &		{\bf 161}	 &0.17 \\\hline
tial    & 14&8 & 581 &   1121  &2.22 & \bf 1060  &2.13 &   1371  &6.68  &	1330	  &5.25 &		1322	  &3.65 \\\hline
vtx1	 & 27& 6&110 &\bf 236	 &0.50  &	247	 &0.51 &		258	 &0.93 &		313	 &0.63 &		317	 &0.62 \\\hline
x7dn	 &66 &15 & 538&	1078	  &1.89 &	1010	  &2.23 &\bf	919	 &2.80  &	1068	 &2.32 &		1043 	 &1.30 \\\hline
\end{tabular}  
 \end{scriptsize}  
  \caption{Comparison of five different variants of the DSOP minimization heuristic (SORT version: {\bf weight/dimension}.) The size of the best DSOP representation computed for each benchmark is in boldface.}
  \label{tab1wd}
  \end{table*}

\subsection{DSOP synthesis}

We have considered the five different variants of the heuristic described in Section~\ref{sec-alg}, denoted as {\bf DSOP-1, DSOP-2, DSOP-3, DSOP-4, DSOP-5}.
For each variant, we have run both versions of the procedure SORT, to estimate  the practical effectiveness of each version. Namely we have ordered the cubes for decreasing dimension and, in case of  equal dimension, for increasing weight (version {\bf dimension/weight}). Then we have ordered the cubes for increasing weight and, in case of  equal weight, for decreasing dimension (version {\bf weight/dimension}).

Since the benchmarks are multi-output functions and the algorithm is described for single output function, in the experiments we have considered each output separately, but the minimization phase  with {\sc espresso} is performed in a multi-output way. Moreover, common disjoint cubes of several output are counted only once.

Tables~\ref{tab1dw} and~\ref{tab1wd} report a significant subset of the experiments. In particular,  Table~\ref{tab1dw} reports the performances of the heuristics with respect to the first version of the SORT procedure, while Table~\ref{tab1wd} is relative to the second SORT procedure. 
All benchmarks in these tables are completely specified.
In both tables, the first column reports the name of the benchmark; the following two columns give the number of inputs and outputs; the column labeled {\bf SOP} shows the number of products in a SOP representation computed by {\sc espresso} in the heuristic mode; finally the remaining five pairs of columns report the number of disjoint products in the DSOP expressions computed by our heuristics and the corresponding synthesis time.

\begin{table}[ht]
  \centering
  \begin{footnotesize}
    \begin{tabular}
{||l||r|r||r|r||r|r||}
\hline
 &   &  & \multicolumn{ 2}{c||}{{\bf DSOP-3 (a)}} & \multicolumn{ 2}{c||}{{\bf DSOP-3 (b)}}\\
 {\bf Bench} &  {\bf in} &{\bf out} &   {\bf size}&  {\bf time}&   {\bf size}&   {\bf time}\\\hline
b10     & 15&11  & 115 &0.49 &115 & 17.04\\\hline
b3      & 32&20  &  279 &1.21 &  279 & 47.34\\\hline
bca&  26& 46& 189	&0.29 &189& 49.54  \\\hline
bcb& 26& 39& 162	&0.26 &		162 &42.33 \\\hline
bench1 & 9& 9&250	&0.32 &\bf 210 &14.92 \\\hline
ex1010  & 10&10 &  876 &1.34 & \bf 665 & 73.00\\\hline
exam     & 10&10 &	145	&0.33 &	\bf 107 & 62.26 	 \\\hline
exep& 30&63 & 130 &0.53 &	\bf	120 & 8.33\\\hline
exps& 8&38 &151&0.31 &		151 &37.91 \\\hline
pdc& 16& 40&381	&0.98 &\bf 277 & 37.14 \\\hline
spla    & 16&46 & 347 &0.64 & 347& 37.06  \\\hline
test2& 11& 35&	2322	 &2.40  &\bf 2054& 324.84\\\hline
test3   & 10&35 &1462 &1.81 &  \bf 1204 &159.60\\\hline
\end{tabular}  
 \end{footnotesize}  
  \caption{DSOP synthesis of incompletely specified benchmarks, without ({\bf DSOP-3 (a)}) and with ({\bf DSOP-3 (b)}) elimination of cubes covering only don't cares. The size of the best DSOP is in boldface.}
  \label{tabDC}
  \end{table}

\begin{table*}
  \begin{center}
  \begin{scriptsize}
   \setlength{\tabcolsep}{5pt} 
\begin{tabular}{||l||r|r||r|r|r|r|r|r|r|r||r||}
\hline
{\bf Bench} &   {\bf in} &  {\bf out} &  {\bf PLA} & {\bf SOP} &$\begin{array}{c}\mbox{\bf DSOP}\\ \mbox{\bf \sc espr.}\end{array}$&$\begin{array}{c}\mbox{\bf DSOP}\\ \mbox{\bf \cite{FSC93}}\end{array}$&$\begin{array}{c}\mbox{\bf DSOP}\\ \mbox{\bf \cite{ST02}}\end{array}$&$\begin{array}{c}\mbox{\bf DSOP}\\ \mbox{\bf \cite{FD02}}\end{array}$&$\begin{array}{c}\mbox{\bf DSOP}\\ \mbox{\bf \cite{DHFD04}}\end{array}$ &$\begin{array}{c}\mbox{\bf DSOP}\\ \mbox{\bf \cite{BE10}}\end{array}$&{\bf DSOP-3} \\\hline
5xp1    &  7 & 10 &   75 &  65 &   99 &  70 &   -- &   82 &79 &{\bf 48}& 70 \\\hline
9sym    &  9 &  1 &   87 &  86 &  209 & 166 &  148 &  148 & 148&--&{\bf 134} \\\hline
alu4    & 14 &  8 & 1028 & 575 & 3551 &  -- &   -- & 1545 & 1372&1206&{\bf 881} \\\hline
b12     & 15 &  9 &  431 &  43 &  691 &  57 &   -- &   60 & 60 &62&{\bf 51} \\\hline
clip    &  9 &  5 &  167 & 120 &  359 & 162 &   -- &  262 &212 &167&\bf 140 \\\hline
co14    & 14 &  1 &   47 & 14 &\bf   14 &  -- &\bf   14 &\bf   14 &  --&--&\bf14 \\\hline
cordic & 23 & 2& 1206& 914&22228 &-- & --&  19763& 8311&\bf 6687&9893\\\hline
inc 		&7	 &9  &      34 & 30&  56  &-- &    --&  66&\bf27&--&37\\\hline
max1024 & 10 &  6 & 1024 & 274 &  775 &  -- &   -- &  444&--&362 &\bf 334 \\\hline
misex1  &  8 &  7 &   32 &  12 &   18 &  \bf 15 &   -- &   34 &34&\bf 15&\bf  15 \\\hline
misex2  & 25 & 18 &   29 &  28 &   29 & \bf 28 &   -- &   30 & 29 &\bf 28&\bf28 \\\hline
misex3  & 14 & 14 &  1848 &  690 &  2349   & --   &   -- &    2255 &1973&--& \bf  1032 \\\hline
mlp4    &  8 &  8 &  256 & 128 &  206 &  -- &   -- &  203 &--&155& \bf143 \\\hline
rd53    &  5 &  3 &   32 &  31 & \bf  31 & \bf 31 &   -- &   35 & 35&\bf31& \bf31 \\\hline
rd73    &  7 &  3 &  141 & 127 & \bf 127 &\bf 127 &   -- &  147&147& \bf127&\bf 127 \\\hline
rd84    &  8 &  4 &  256 & 255 & \bf 255 &  -- &   -- &  294 &294&\bf255&\bf 255 \\\hline
sao2   & 10 &  4 &  58 & 58 &  199 &  -- &  -- &  96 &96&--& \bf24 \\\hline
sym10   & 10 &  1 &  837 & 210 &  367 &  -- &  240 &  240&--& --& \bf232 \\\hline
t481    & 16 &  1 &  481 & 481 & 2139 &  -- & 2139 & 1009 &\bf841&--&\bf 841 \\\hline
x7dn    & 66 & 15 &  622 & 538 & 1697 &  -- &   -- & 1091 &--&1228& \bf812 \\\hline
xor5    &  5 &  1 &   16 &  16 &   \bf 16 &  -- &   \bf  16 &   \bf 16 &\bf 16&\bf 16&\bf 16 \\\hline
\end{tabular}
\end{scriptsize}  
 \end{center}
  \caption{Comparison with other techniques. The size of the best DSOP is in boldface.}
  \label{tab2}
\end{table*}

As Table~\ref{tab1dw} and Table~\ref{tab1wd} clearly show, the third variant of the heuristic, together with the first version of procedure SORT (version {\bf dimension/weight}),  gives the best results regarding the size of the resulting DSOP forms, and its running times are comparable to those of the other variants, and sometimes even lower.

We have then tested the performances of the best variant of our heuristic on incompletely specified benchmarks. Table~\ref{tabDC} reports a subset of our experiments. We have run the heuristic without the elimination of cubes covering only don't cares points from the solution under construction ({\bf DSOP-3 (a)}), and with such elimination ({\bf DSOP-3 (b)}). As the table clearly shows,  the elimination of these cubes naturally produces  better solutions in terms of size,  but the computational time is much higher. 

In another series of experiments we compared   our heuristic  (with the third version of the optimization phase, and without elimination of cubes of don't cares only) with other DSOP minimization methods. 
We considered three techniques working, as ours, on explicit representation of cubes, and one method based on binary decision diagrams.
The first algorithm~\cite{FSC93} sorts cubes in a minimal SOP according to their size, and compares the largest cube with all the others, starting from the smallest ones. In the next step, the second largest cube is selected and compared to all smaller ones, etc. As a last step, the cubes are merged wherever possible.
The second algorithm, presented in~\cite{ST02}, exploits the property of the most binate variable in a set of cubes to compute a DSOP form.
The algorithm proposed in~\cite{BE10} enumerates all overlapping pairs of cubes in a SOP form, and builds a disjoint cover starting from the pairs of cubes with the highest degree of logic sharing.

Finally, the third approach, presented in~\cite{FD02}, makes use of BDDs, exploiting the efficiency resulting from the implicit representation of the products. Observe in fact that a DSOP form can be extracted in a straightforward way from a BDD, as different one-paths correspond to disjoint cubes. As the results presented in~\cite{FD02} largely depend on the variable ordering of the underlying BDD, in~\cite{DHFD04} an evolutionary algorithm has been proposed to  find an optimized variable ordering for the BDD representation that guarantees more compact DSOP forms.

Table~\ref{tab2}  reports a cost-oriented comparison among the different methods. The first three columns are as before. Columns four and five report the number of products in the PLA realization and in the SOP form heuristically minimized by {\sc espresso} in the heuristic mode. The column labeled 
{\bf DSOP {\sc espr.}} shows the size of the DSOP computed running {\sc espresso} with the option ``-Ddisjoint'' on the previously computed SOP form. The next five columns report the sizes, when available, of the DSOP forms computed with the methods discussed in~\cite{FSC93},~\cite{ST02},~\cite{FD02},~\cite{DHFD04}, and~\cite{BE10}, respectively. Finally, the last column shows the size of the DSOPs computed with our heuristic (third variant).

As the table clearly shows, our method almost always generates smaller DSOP representations, and the gain in size can be quite striking, as for instance for the benchmarks {\em alu4, clip} and {\em misex3}. We have found only a few benchmarks where our approach compares unfavorably: {\em 5xp1, cordic} and {\em inc}.

A time comparisons among all these different methods was not possible due to the partial absence of CPU times specification in the literature.

\subsection{Partial DSOP synthesis}
In order to test our partial DSOP synthesis algorithm, we have applied the heuristic to the classical {\sc espresso} benchmark suite~\cite{Y91} with the following meaning. We have considered only benchmarks with don't cares, where the on set of the benchmark is the on set of $sopD$, and the don't care set of the benchmark is the don't care set of $sopS$. 

Table~\ref{tabPDSOP} reports a subset of our experimental results. 
The column labeled {\bf SOP} shows the number of products in a SOP representation computed by {\sc espresso} in the heuristic mode. The remaining three pairs of columns report the number of  products  and the corresponding synthesis time for the following three forms (all computed with the third version of the optimization phase, the dimension/weight sort version, and with the elimination of cubes covering don't cares only):
\begin{enumerate}
\item {\bf DSOP}: a DSOP for the original function, with the choice of don't cares performed by  {\sc espresso} in the heuristic mode. Each don't care point is covered {\bf at most} once.
\item {\bf P-DSOP (a)}: a partial DSOP for the original function, with the choice of don't cares performed by {\sc espresso} in the heuristic mode. Don't care points are either eliminated or covered {\bf at least} once.
\item {\bf P-DSOP (b)}: a partial DSOP for the original function, where all the don't cares of the function are in play (they have all been covered during the first SOP minimization).
Don't care points are either eliminated or covered {\bf at least} once  in the final form.
\end{enumerate}
Note that the results in the column {\bf SOP} are better than ours because the resulting form is not disjoint.

The table suggests that the best solution is the one relative to the choice of don't cares made by  {\sc espresso}.
Moreover, it appears clearly from these results that the option of covering more than once the don't care points of the function ({\bf DSOP-3 (a)}) gives better results, especially for big benchmarks.

\begin{table*}[ht!]
  \centering
  \begin{footnotesize}
    \begin{tabular}
{||l||r|r||r||r|r||r|r||r|r||}
\hline
  &   &  & {\bf SOP} &\multicolumn{ 2}{c||}{{\bf DSOP}} & \multicolumn{ 2}{c||}{{\bf P-DSOP (a)}} & \multicolumn{ 2}{c||}{{\bf P-DSOP (b)}} \\
 {\bf Bench} &  {\bf in} &{\bf out}  &{\bf size} &{\bf size}&{\bf time}&{\bf size} &{\bf time}&{\bf size} &{\bf time} \\ \hline
alu3&10&8&66& 67 & 4.65&	67& 10.83& 67 & 13.06\\\hline
apla&10&12&25&33 & 4.11&\bf	25 &2.49&\bf 25& 13.204\\\hline
b10	&15&11&100&\bf 115 & 17.04&\bf 115&15.96&117 & 18.22\\\hline
b3&32&20&211& 279 &47.34  &279& 51.39 &279 & 54.80\\\hline
b4&33&23&54&	 62&10.06& 62& 5.99& 62& 7.74\\\hline
bca&26& 46& 180&\bf	 189& 49.54&\bf  189& 33.63& 190& 53.11\\\hline
bcb	&26&39&155&162 &42.33&162 & 29.08&162 &42.18\\\hline
bcc	&26&45&137& 145& 26.08&145 & 31.19&145 & 43.24\\\hline
bcd	&26&38&117&121 & 17.30&121& 20.51&121& 29.14\\\hline
bench1&9&9&139&	210 & 14.92&\bf 164&18.30&246 & 86.90\\\hline
dk17	&10&11&18&22& 1.32&\bf 19&1.50&\bf 19&15.57\\\hline
dk27	&9&9&10&12 & 1.04&\bf 10&0.73&\bf 10&10.15\\\hline
dk48	&15&17&22&24& 1.17&\bf 22&1.40&\bf 22&69.59\\\hline
duke2&22&29&12&	26& 3.06&24&4.10&\bf 23&26.16\\\hline
ex1010	&10&10&284&665&73.00&\bf 481&87.24&739&282.44\\\hline
exam	&10&10&67&107	&62.26&\bf 89&66.78&168&115.78\\\hline
exp	&8&18&59&72&7.11&70&7.28&\bf 63&16.25\\\hline
exps	&8&38&136&\bf 151&	37.91&\bf 151&41.37&152&5.61\\\hline
inc& 7& 9& 30&\bf 37&2.52 &38&3.298&41&4.87\\\hline
mark1&20&31&19&	29& 2.71&\bf 23&5.26&25&136.14\\\hline
p1&8&18&55&	90&	7.51&\bf 67&12.26&79&32.37\\\hline
p3&8&14&39&	71&	4.77&\bf 47&10.29&52&18.27\\\hline
pdc&16&40&145&	277&	 37.14&203&68.65&\bf 191&291.51\\\hline
sao2	&10&4&9&24&	1.09&20&4.01&\bf 19&6.86\\\hline
spla	&16&46&260&347&37.06&347&50.05&347&53.21\\\hline
t2&17&16&53&	 \bf 58&3.34&59&5.36&59&8.48\\\hline
t4	&12&8&16&18	&1.31&\bf 17&1.79&21&10.49\\\hline
test1&8&10&121&	169	&10.52&\bf 141&12.92&197&45.38\\\hline
test2&11&35&1103&	2054	&324.84&\bf 1354&361.73&1996&2921.75\\\hline
test3&10&35&541&	1204	&159.60&\bf 746&179.11&1001&1405.99\\\hline
test4&8&30&120&	466	&65.34&335&82.44&\bf 278&253.03\\\hline
x1dn	&27&6&70&148	&12.51&\bf 98&26.31&106&48.30\\\hline
\end{tabular}  
 \end{footnotesize}  
  \caption{Partial DSOP synthesis, using the subset of don't cares selected by {\sc espresso-non-exact} ({\bf P-DSOP (a)}) or all don't cares of the original function ({\bf P-DSOP (b)}).}
  \label{tabPDSOP}
  \end{table*}

\section{Conclusions and Future Work}
\label{sec-fut}

Deriving an optimal DSOP or partial DSOP
representation of a Boolean function is a hard problem. This is why we have proposed a heuristic  that has been implemented, tested, and compared with others.

From the experimental results we conclude that exploiting SOP minimization for DSOP synthesis is a crucial idea. In fact, comparing our results with the ones in the literature we always obtain equal or smaller forms.
We observe that the fact that SOP and DSOP problems are so close is not intuitive. In fact, we would have expected that efficient strategies to solve the two problems would be different since DSOP minimization appears to be much harder then SOP synthesis. Nevertheless, the experiments show that, starting from minimal or quasi-minimal SOP expressions, we can heuristically derive very compact DSOP forms. Moreover, from Table~\ref{tabPDSOP} we also infer that the choice of the don't cares, which are used as ones of the function, performed for the SOP minimization is nearly always the best choice also for DSOP synthesis.
Therefore it would be interesting to further study the closeness of SOP and DSOP minimal forms both in theoretical and experimental way.

%


It could also be worth studying the approximability of DSOP minimization with the aim of designing  approximation algorithms instead of heuristics.  In fact, while a $p$-approximation algorithm yields a near-optimal solution, i.e. a solution whose cost $C$ is $\leq p C^{\ast}$ where $C^{\ast}$ is the cost of an optimal solution~\cite{GJ79}, no prediction can be made on the result of a heuristic. Perhaps a first step in this direction would be understanding when our heuristic returns a DSOP whose cost is much higher then the cost of an optimal DSOP.




\end{document}